\documentclass{amsart}
\usepackage{hyperref}
\usepackage{graphicx}
\usepackage{curves}
\usepackage{enumerate}
\usepackage{tikz}
\newtheorem{theorem}{Theorem}[section]
\newtheorem{lemma}[theorem]{Lemma}

\newtheorem{remark}[theorem]{Remark}
\newtheorem{proposition}[theorem]{Proposition}

\newcommand{\R}{\mathbb{R}}
\newcommand{\Z}{\mathbb{Z}}

\newcommand{\C}{\mathbb{C}}
\newcommand{\T}{\mathbb{T}}

\newcommand{\noprint}[1]{}
\newcommand{\nn}{\nonumber}
\newcommand{\beq}{\begin{equation}}
\newcommand{\eeq}{\end{equation}}
\newcommand{\bea}{\begin{eqnarray}}
\newcommand{\eea}{\end{eqnarray}}

\newcommand{\ol}{\overline}

\newcommand{\ti}{\tilde}

\newcommand{\disp}{\displaystyle}

\newcommand{\clos}{\mathop{\mathrm{clos}}}

\newcommand{\id}{\mathbb{I}}
\newcommand{\I}{\mathrm{i}}
\newcommand{\E}{\mathrm{e}}

\newcommand{\re}{\mathop{\mathrm{Re}}}
\newcommand{\im}{\mathop{\mathrm{Im}}}

\DeclareMathOperator{\res}{Res}
\DeclareMathOperator{\Ai}{Ai}

\makeatletter
\newcommand{\dlmf}[1]{%
\cite[%
  \def\nextitem{\def\nextitem{, }}%
  \@for \el:=#1\do{\nextitem\href{http://dlmf.nist.gov/\el}{(\el)}}%
]{dlmf}%
}
\makeatother

\newcommand{\arxiv}[1]{\href{http://arxiv.org/abs/#1}{arXiv:#1}}
\newcommand{\eps}{\varepsilon}
\newcommand{\vphi}{\varphi}
\newcommand{\si}{\sigma}
\newcommand{\la}{\lambda}
\newcommand{\ga}{\gamma}

\numberwithin{equation}{section}

\newcommand{\ra}{\begin{pmatrix}  \alpha & \alpha \end{pmatrix}}
\newcommand{\rn}{\begin{pmatrix}  \eta & \eta \end{pmatrix}}

\begin{document}

\title[The parametrix problem]{The parametrix problem for Toda equation with steplike initial data}

\author[A. Pryimak]{Anton Pryimak}
\address{Faculty of Mathematics and Computer Sciences\\ V.N. Karazin Kharkiv National University\\ 4, Svobody sq.\\ 61022 Kharkiv\\ Ukraine}
\email{\href{mailto:pryimakaa@gmail.com}{pryimakaa@gmail.com}}

\maketitle

\section{Introduction}
 The  Toda rarefaction problem is related to the analysis of the long-time asymptotic behaviour of the Cauchy problem solution   for the doubly infinite Toda lattice
\begin{align} \label{tl}
	\begin{split}
\dot b(n,t) &= 2(a(n,t)^2 -a(n-1,t)^2),\\
\dot a(n,t) &= a(n,t) (b(n+1,t) -b(n,t)),
\end{split} \ \qquad (n,t) \in \Z \times \R_+,
\end{align}
with symmetric steplike initial data
 $a(-n,0)=a(n,0)$, $b(-n,0)=-b(n,0)$, such that
 $a(n,0)\to\frac{1}{2}$ and $b(n,0)\to\pm \hat b$, $\hat b>0$, as $n\to\pm\infty$. From a physical point of view it is interesting to study  asymptotics of the solution in the regime when $t\to +\infty$, $n\to\infty$ with the ratio $\xi:=\frac{n}{t}$  slow varying. The Toda rarefaction  wave demonstrates qualitatively  different behaviour  depending on the value of  the background constant $\hat b$, where we distinguish between the cases $\hat b> 1$ and  $0<\hat b\leq 1$. The behaviour of the rarefaction wave depends also on a value of $\xi$ varying in some intervals of $\R$.  In particular, for $\hat b>1$ one can observe four sectors  with different asymptotics of the solution on the $(n,t)$ half plane. These sectors are divided by the rays corresponding to the leading and the back wave fronts ($\xi=\pm 1$), and to the ray $\xi=0$.  In the regions ahead of the leading wave front and behind the back wave front, an application of the  Inverse Scattering Transform (IST) method results straightforward in  adjusting soliton asymptotics of the solution on the respective backgrounds, while in two middle sectors  the classical IST does not lead to desirable results. On the other hand, in the middle sectors  the Nonlinear Steepest Descent  (NSD) method, developed in \cite{dz} for modified KdV equation,  proves to be more efficient.  For the Toda lattice this method was pioneered by Deift at all in a \cite{dkkz}.  In this seminal work the  NSD approach was applied for the first time for a vector form of the oscillating Riemann-Hilbert  problem (RHP) resulting in establishing   asymptotics  in a  physically important transitional region near the ray  $\xi:=\frac{n}{t}\sim 0$, as $t\to + \infty$.
Recall that NSD approach proceeds by a sequence of transformations (conjugations/deformations) that convert the original RHP into an equivalent RHP (RHP -{\it  equiv}) with a jump matrix $v^{eqv}$ of the form $v^{eqv} = v^{mod} + v^{err}$, where $v^{mod}$ is a jump matrix for  an explicitly solvable RHP (RHP-{\it mod}), and the entries of  an error matrix $ v^{err}$ are exponentially small with respect to $t$ except of a finite number of small vicinities of critical (parametrix) points. Solving the RHP corresponding to $v^{mod}$ yields the principal term of asymptotic expansion of the solution with respect to large $t$. To estimate the error term one has to  rescale the  RHP - {\it equiv } in vicinities of the parametrix points and solve respective  local RH problems. One has to mention that for steplike solutions a contribution of the parametrices in asymptotics  is perceptible only for the second or even the third term of the expansion.
By means of the same NSD  approach for a vector RHP in \cite{EM}, the long - time asymptotic behaviour of the solution
was studied in all main regions of $(n,t)$ half-plane  for more general initial data:
\begin{align} \label{ini1}
\begin{split}	
& a(n,0)\to a, \quad b(n,0) \to b, \quad \mbox{as $n \to -\infty$}, \\
& a(n,0)\to \frac{1}{2} \quad b(n,0) \to 0, \quad \mbox{as $n \to +\infty$},
\end{split}
\end{align}
where $a>0$ and $b\in\R$ satisfy the condition
\beq\label{main} 1 < b-2a.
\eeq
By an analogy  the problem \eqref{tl}-\eqref{main} is also called the Toda rarefaction problem.
Note that the initial value problem \eqref{tl}--\eqref{ini1} is uniquely solvable for any constants $a>0$, $b\in \R$ and any initial data which approach their limiting constants with a polynomial rate.  Moreover,  for each   $t\neq 0$  the solution tends    as $n\to\pm\infty$ to the same  constants, and with the same  rate as the initial data (cf.\ \cite{emt3}). However, an application  of NSD approach require a faster speed of approximation, namely, in \cite{dkkz} and in \cite{EM} it was assumed that
\beq \label{decay}
 \sum_{n = 1}^{\infty} \E^{\nu n} \big( |a(-n,0) - a| + |b(-n,0)-b|
+ |a(n,0) - \tfrac{1}{2}| + |b(n,0)|\big) < \infty,
\eeq
for some  $\nu>0$. Again, one can expect (see \cite{m15} ) that  the long-time asymptotics of the solution for  \eqref{tl}--\eqref{ini1} are determined qualitatively by the mutual location  of the intervals $[b-2a, b+2a]$ and $[-1,1]$, and by the discrete spectrum $\lambda_1,...,\lambda_N$ of the underlying Jacobi operator
  \begin{equation}\label{ht}H(t)y(n):=a(n-1,t)y(n-1) + b(n,t)y(n) + a(n,t)y(n+1),\quad n\in\mathbb Z.
\end{equation}
In particular, in \cite{EM}  it was shown that   for $t\to +\infty$ the solution $\{a(n,t), b(n,t)\}$ of the problem \eqref{tl}-\eqref{decay}:
\begin{itemize}
\item In the region $n>t$ is asymptotically close to the right constants
 $\{\frac{1}{2}, 0\}$ plus a sum of solitons corresponding to the eigenvalues $\la_j<-1$.
\item In the region $0<n<t$:
\beq\label{ext}
a(n,t)=\frac{n}{2t} +O\Big(\frac{1}{t}\Big),
\quad b(n,t)= 1 - \frac{n}{t}
 +O\Big(\frac{1}{t}\Big).
\eeq
\item In the region $-2a t<n< 0$:
\beq\label{ext2}
a(n,t)=-\frac{n}{2t}+ O\Big(\frac{1}{t}\Big), \quad b(n,t)= b- 2a - \frac{n}{t}+ O\Big(\frac{1}{t}\Big).
\eeq
\item In the region $n < -2a t$, the solution of \eqref{tl}--\eqref{decay} is close
to the left background constants $\{a, b\}$ plus a sum of solitons corresponding to the eigenvalues $\la_j>b+2a$.
\end{itemize}
By an analogy with the KdV rarefaction waves (cf. \cite{aelt},\cite{LN},\cite{LN1}), in \cite{EM} it was conjectured that the error terms $O(t^{-1})$ are
uniformly bounded with respect to $n$  for $\varepsilon t\leq n\leq (1-\varepsilon)t$ in \eqref{ext}, and for $(-2 a+
\varepsilon)t \leq n\leq -\varepsilon t$ in \eqref{ext2}, where $\varepsilon>0$ is an arbitrary small value. Moreover, it was conjectured that the influence of parametrices is not perceptible for the first two terms of the asymptotic expansion in the two middle regions. More detailed  analysis of transformations from original RHP to RHP-{\it equiv} allowed then to derive  a precise formula for the second term of asymptotic expansion.

Our paper is a continuation of \cite{EM}. We solve rigorously the parametrix problems associated with the region $0<n<t$ and complete the asymptotic analysis, justifying the asymptotics obtained in \cite{EM}. In particular, we show that the parametrix problems solutions do not contribute indeed in the first two terms of asymptotic expansion. Note, that in \cite{EM} it was assumed that  points $b-2a$ and $b+2a$ are nonresonant. In this paper resonances are admitted there. The presence of a resonance at the edge of the spectrum of operator \eqref{ht} implies a non $L^2$ singularity of the jump matrix in the original RHP. In turn it requires an additional discussion of the statement of the RHP, of the solution uniqueness, and modifications in the transformation steps which lead to RHP-{\it equiv}.  We investigate   the region $0<n<t$
only. The asymptotical analysis in the region $-2a t<n< 0$ is sequent if one considers the Toda lattice solution \[\hat a(n,t)=\frac{1}{2a} a(-n-1, \frac{t}{2a}), \quad \hat b(n,t)=\frac{1}{2a}\left(b - b(-n,\frac{t}{2a})\right).\]
This solution corresponds to the initial profile
\[\hat a(n,0)\to \frac{1}{2a},\ \ \hat b(n,0)\to\frac{b}{2a}, \ \mbox{as}\ \  n\to -\infty,\]\[\hat a(n,0)\to\frac{1}{2},\ \ \hat b(n,0)\to 0, \ \mbox{as}\ \  n\to +\infty,\]
and the asymptotics of it in the region $0<n<t$ immediately implies the asymptotics of solution for \eqref{tl}-\eqref{ini1} in the region $-2a t<n< 0$.
As for the   soliton regions, they are studied rigorously  in \cite{KTb} for the decaying case. In the steplike case the analysis is very similar.
The paper has the following structure. In section \ref{rhpstat} we recall some necessary facts from scattering theory for Jacobi operator with steplike backgrounds and study  the unique solvability of the initial vector RHP, taking into account a possible presence of resonances. In section \ref{sec:steps} we lists some conjugation/deformation invertible transformations given in \cite{EM} which reduce the initial meromorphic RHP to a holomorphic RHP with jumps close as $t\to\infty$ to constant matrices except of vicinities of two points, where we pose and solve the parametrix problem (Section \ref{parametrix}). In section \ref{asymptotics} we perform a completion of the asymptotic analysis related to the Cauchy type integrals and singular integral equations.

\section{Statement of the Riemann-Hilbert problem}\label{rhpstat}
Let $\{a(n,t), b(n,t)\}$ be the solution for the initial value problem \eqref{tl}-\eqref{decay}and let $H(t)$ be the Jacobi operator \eqref{ht}. Consider the underlying spectral problem:
 \beq\label{ht1}
H(t)y(n)=\la y(n), \ \la\in\C.
\eeq
Introduce also the left and the right background operators:
\beq\label{backgr}\aligned
H\,y(n)&:=\frac{1}{2}y(n-1)  + \frac{1}{2} y(n+1),\quad n\in\Z,\\
H_1\, y(n)&:=ay(n-1) + by(n) + ay(n+1),\quad n\in\Z.
\endaligned\eeq
Under  condition \eqref{main} they have disjoint spectra $\si(H)=[-1,1]$ and $\si (H_1)= [b-2a, b+2a]$ with the mutual location as depicted in Fig.~\ref{fig:specloc}.
 \begin{figure}[h]
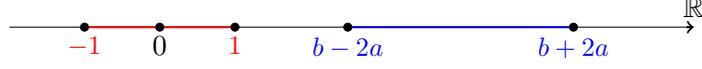

\tikz{
%\draw[very thin] (-2,0)--(5,0);
\draw (-2,0)--(7,0);
\draw[thick](7,0)--(7.1,0)[->]node[above]{\large $\R$};
\node[below] (0,0) {\large ${0}$};
\draw[thick,red] (-1,0)node[below]{ ${-1}$}--(1,0)node[below]{ ${1}$};
\draw[thick,blue] (2.5,0)node[below]{ ${b-2a}$}--(5.5,0)node[below] {${b+2a}$};
\filldraw[black] circle(1.5pt) (-1,0)circle(1.5pt)
 (1,0)circle(1.5pt) (2.5,0)circle(1.5pt) (5.5,0)circle(1.5pt) ;
}
\caption {Mutual location of background spectra}\label{fig:specloc}
\end{figure}

 Operator $H(t)$ has the continuous spectrum of multiplicity 1, consisting of the union of the background spectra plus a finite number of eigenvalues
\[\{\lambda_j\}_{j=1}^{N} \subset \R \setminus ([-1,1]\cup [b-2a, b+2a]).\] Instead of the spectral parameter $\la$ we use it's Joukowsky transform $z$:
\beq\nn
\la=\frac{1}{2}\left(z + z^{-1}\right),\quad |z|\leq 1.
\eeq
Denote
\beq\label{defqu}
q_1:=z(b-2a),\quad q_2:=z(b+2a), \quad I:=[q_2,\,q_1].
\eeq
where $z(\la)=\la-\sqrt{\la^2-1}$. Note that the map  $z\mapsto \la$ is a bijection between the sets  $\mathfrak D:=(\mathbb D\setminus I)$ and $\C\setminus([-1,1]\cup[b-2a,b+2a]))$ , where $\mathbb D:=\{z: |z|<1\}$. At the same time $z\mapsto \la$ is a bijection between the sets $\ol{\mathfrak D}:=\clos \mathfrak D$ and $\clos(\C\setminus([-1,1]\cup[b-2a,b+2a]))$, if we treat  the closure as adding to boundaries the points of the upper and lower sides along the cuts, while considering them as distinct points. In particular, points of  $I\pm \I 0$ are treated as different points. We will call the points $z_j:=z(\la_j),\quad j=1,\dots,N$ discrete spectrum of the operator $H(t)$.
Introduce also the Joukowsky transform associated with the left background operator $H_1$:
\[\la=b + a\left(\zeta + \zeta^{-1}\right),\quad |\zeta|\leq1,\] then the function $\zeta(z)$ is a single-valued analytic function in $\mathfrak D$, continuous up to the boundary.

Consider now the Jost solutions  of equation \eqref{ht1}  which  are asymptotically close to the free exponents of background operators \eqref{backgr}:
\beq\label{defJost}
\lim_{n\to \infty} z^{-n}\psi(z,n,t) =1, \quad
\lim_{n\to -\infty} (\zeta(z))^{n}\psi_1(z,n,t) =1, \qquad \forall z \in \ol {\mathfrak D}.
\eeq

Denote by $W(z,t):=\langle\psi_1,\psi\rangle$ the Wronskian of the Jost solutions. Here \[\langle f, g\rangle:=a(n-1,t)\left(f(n-1) g(n) - g(n-1)f(n)\right).\] As a function of z, the Wronskian $W(z,t)$  is a holomorphic function in $\mathfrak D $, continuous up to the boundary. In $\mathfrak D$ it has simple  zeros at the points  $z_j$. The Jost solutions are real valued and dependent at these points. By \eqref{defJost} they are  eigenfunctions of $H(t)$. Denote
\[\ga_j(t):=\Big(\sum \limits_{n\in\Z}\psi^2(z_j,n,t)\Big)^{-2}.\] Next, on the boundary $\partial \mathfrak D$ the Wronskian  can vanish only on the set $\{-1,1,q_1,q_2\}$.  By definition, a point $p \in\{-1,1,q_1,q_2\}$ is called a resonant point if $W(p,0)=0$. In this case \beq\label{Wreso}
W(p,t) = C(t)\sqrt{z-p}\,(1+o(1)),\, \mbox{as} \, \, z \to p,\quad C(t)\ne 0\ \ \forall t\geq 0.
\eeq
Consider now the scattering relations
 \beq \nn
 T(z,t)\psi_1(z,n,t)=\ol {\psi(z,n,t)} + R(z,t)\psi(z,n,t), \quad |z|=1,
\eeq
where $T(z,t)$ and $R(z,t)$ are the right  transmission and reflection coefficients. Since $T(z,t)=(z-z^{-1})(2W(z,t))^{-1}$, the transmission coefficient can be continued as a meromorphic function in $\mathfrak D$. Due to \eqref{Wreso}  it has continuous limiting values on $\mathbb T$ and on the sides of the interval $(q_2, q_1)$. Moreover, its modulo does not have a jump on $I$.
Define now the function
\beq \label{defchi}
\chi(z):=\frac{2a\big(\zeta(z-\I 0)  - \zeta(z-\I 0)^{-1}\big)}{z^{-1}-z}|T(z,0)|^2, \quad z\in I.
\eeq
 The function $\chi(z)$ is continuous on the interval $ I$ except of possibly the endpoints \eqref{defqu}.  If $q_i$ is a resonant point then
\beq\label{ChiRes}
\chi(z) = C(z- q_i)^{-1/2}\big(1+o(1)\big), \,\, C \ne 0, \quad z \to q_i,
\eeq
otherwise  $\chi(q_i) = 0$ (in the non-resonant case).

 Denote now
$
R(z):=R(z,0), \quad  \ga_j:=\ga_j(0).
$
As is known (cf. \cite{dkkz},\cite{vdo},\cite{EMT7},\cite{EMT8}), the IST approach allows us to restore uniquely the solution of \eqref{tl}-\eqref{main} from a  minimal set of the initial scattering data
\beq\label{SET}\{R(z),\  z\in\T,\ \  \chi(z), \ \ z\in I, \ \ \la_j, \gamma_j,\ \ j=1,..., N.\}\eeq This is  the data  which are involved in the right Marchenko equation for the step-like case \eqref{ini1}-\eqref{main}. Note that the time-dependent Marchenko equation encloses  the values $R(z,t)$, $|T(z,t)|$ (cf. \eqref{defchi}) and $\gamma_j(t)$, $j=1,...,N$,  whose   evolution  due to the Toda flow is given by:
$\gamma_j(t)=\gamma_j(0)\exp((z_j-z_j^{-1})t)$,
\[
R(z,t) = R(z)\E^{(z-z^{-1})t},\, z \in \T; \quad |T(z,t)|^2 = |T(z)|^2\E^{(z-z^{-1})t},\, z \in I.
\]

In $\mathfrak D$ introduce a vector-function $\quad m(z)=\begin{pmatrix}m_1(z,n,t),& m_2(z,n,t)\end{pmatrix}:$
\beq \label{defm}
m(z,n,t) =
\begin{pmatrix} T(z,t) \psi_{1}(z,n,t) z^n,  & \psi(z,n,t)  z^{-n} \end{pmatrix}.
\eeq
The space and time variables are treated here as parameters. We omit them in the notation of $m$ whenever it is possible.
\begin{lemma}[\cite{emt14}]\nn
The components of the vector function \eqref{defm} satisfy
\[
m_1(z)  = \prod_{j=n}^\infty 2a(j,t) \Big(1 + 2 z \sum_{m=n}^\infty b(m,t)\Big) + O(z^2),\]
\[m_1(z)\,m_2(z)=1+  O(z^2), \quad \mbox{as}\  z \to 0.\]
\end{lemma}
The first component  $m_1(z)$ is a meromorphic function in $\mathfrak D$ with poles at $z_j$, the second one is a holomorphic function continuous up to the boundary. We extend $m$ to the set $\mathfrak D^*:=\C\setminus\left(\,\ol{\mathbb D} \cup I^*\right)$ where $I^*:=[q_2^{-1}, q_1^{-1}]$,
by the following symmetry condition $m(z^{-1}) = m(z) \si_1$, where $\si_1= {\scriptsize \begin{pmatrix} 0  & 1 \\
1 & 0 \end{pmatrix}}$ is the first Pauli matrix. With this extension, the second component $m_2(z)$ is a meromorphic function in $\C\setminus\left(\,\ol{\mathbb D} \cup I^*\right)$ with poles at the points $z_j^{-1}$.

Endow the circle $\T$ with counterclockwise orientation and intervals $I, I^*$ with orientation towards the origin of the circle $\T$. As always a positive side of a contour is the one
which lies to the left as one traverses the contour in the direction it is oriented; the negative side lies to the right respectively. We will use a common abbreviation $m_\pm (z)$ which denotes the limit of $m(z)$ from the positive/negative side respectively of the contour $\Gamma:=\T\cup I\cup I^*$. Using these abbreviations we implicitly assume the existence of respective limits. For $m(z)$ these limits are continuous on $\Gamma$ with the only  possible  singularities  at  points $q_i$ and $q_i^{-1}$, $i=1,2$.

\begin{proposition}
Suppose that the initial data of the
Cauchy problem \eqref{tl}--\eqref{main}, satisfy \eqref{decay}.
Let the set \eqref{SET}  be
the right scattering data of the operator $H(0)$ \eqref{ht}. Then the vector-valued function
$m(z)=m(z,n,t)$ defined by \eqref{defm} is the unique solution of the following vector Riemann--Hilbert problem:
Find a meromorphic in $\C\setminus\Gamma$ function $m(z)$ with poles at the points $z_j$, $z_j^{-1}$,
 which satisfies:
\begin{enumerate}[I.]
	\item  The jump condition $m_{+}(z)=m_{-}(z) v(z)$,  where
\[
v(z)=\left\{
\begin{array}{ll}
\begin{pmatrix}
0 & - \ol{R(z)} \E^{- 2 t \Phi(z)} \\
R(z) \E^{2 t \Phi(z)} & 1
\end{pmatrix}, & z \in \T,\\[3mm]
\begin{pmatrix}
1 & 0 \\
\chi(z) \E^{2t\Phi(z)} & 1
\end{pmatrix}, & z \in I,\\[3mm]
\begin{pmatrix}
1 & -\chi(z^{-1}) \E^{-2t\Phi(z)} \\
0 & 1
\end{pmatrix}, & z \in  I^*.
\end{array}\right.
\]
Here  $\chi(z)$ is given by \eqref{defchi} and {\bf the phase function} $\Phi(z)=\Phi(z, n, t)$ is defined by
\beq  \nn
\Phi(z) = \Phi(z,\xi)=\frac{1}{2} \big(z - z^{-1}\big) + \xi\log z, \quad \xi:=\frac{n}{t}\in \R.
\eeq
\item
The symmetry condition
\beq \label{sc}
m(z^{-1}) = m(z) \si_1.
\eeq
\item
The normalization condition
\beq\label{eq:normcond}
m_1 (0)\, m_2(0)=1, \quad
\quad m_1(0) > 0.
\eeq
\item
The pole condition
\beq \label{polecond}\aligned
\res \limits_{z_j} m(z) &=
( - z_j \gamma_j m_2(z_j)\E^{2 t\Phi(z_j)},\ \  0 );\\
\res \limits_{z_j^{-1}} m(z) &= (0,\ \
 z_j^{-1}m_1(z_j^{-1}) \gamma_j \E^{2 t\Phi(z_j)}).
\endaligned\eeq
\item
 If a point $q_i$  \eqref{defqu} is the resonant point, that is if $\chi(z)$ satisfies \eqref{ChiRes}, then \[
m(z) = \begin{pmatrix}C_1(z-q_i)^{-1/2}, & C_2 \end{pmatrix}\big(1 + o(1)\big),\quad C_1 \ne 0, \quad z \to q_i,
\]
with an analogous singularity of the second component $m_2(z)$ at the point  $q_i^{-1}$.

 If $\chi(q_i)=0$ then $m(z)$ has  limiting values as $z\to q_i$, $z\in\mathfrak D$ and $z\to q_i^{-1}$
$z\in\mathfrak D^*$.
\end{enumerate}
\end{proposition}

This proposition is proved in Appendix.

Note that the symmetry condition II plays a crucial role in ascertaining of the solution uniqueness  for the vector RHP. That is why we will perform only those transformations which  will  preserve this symmetry and also the normalization condition IV. To this end we introduce  a few constraints on our conjugation/deformation steps. The original RHP I-V is in agreement with these constraints already. Namely, suppose that after some steps we got an equivalent RHP for a vector function $\tilde m$ with a jump matrix $\tilde v$ on a contour $\tilde\Gamma$. Then:
\vskip 2mm
\noindent
(1) The jump contour $\tilde\Gamma$ is symmetric with respect to the map $z \mapsto z^{-1}$.
\vskip 2mm
\noindent
(2) The direction on parts of $\tilde\Gamma$ are chosen in a way that the jump matrix satisfies the symmetry condition
\beq\label{matsym}(\tilde v(z))^{-1}=\sigma_1\tilde v(z^{-1})\sigma_1.
\eeq
\vskip 2mm
\noindent
(3) The vector $\tilde m$ satisfies the symmetry condition $\tilde m(z^{-1})=\tilde m(z)\sigma_1$ for all $z\in\mathbb C\setminus\tilde\Gamma$, moreover, $\tilde m_1(0)\tilde m_2(0)=1$.
\vskip 2mm
\noindent
(4) Let $\Gamma^\prime\subset \tilde\Gamma$ be a symmetric sub-contour and let $d: \C\setminus  \Gamma^\prime\to\C$ be a sectionally analytic function. Suppose that $d(z^{-1}) = d(z)^{-1}$ for $z \in \C\setminus \Gamma^\prime$ and $d(0)>0$. Then the conjugation \beq\label{diss}\hat m(z)=\tilde m(z) [d(z)]^{-\sigma_3},\ \ \mbox{where}\ \ [d(z)]^{- \si_3} := \begin{pmatrix} d^{-1}(z) & 0 \\
0 & d(z) \end{pmatrix},\eeq
is in agreement with  the constraints above. To simplify further consideration note that  transformation \eqref{diss} convert the jump matrix $\tilde v$ into a jump matrix
\[
\hat{v} =
\left\{
\begin{array}{ll}
  \begin{pmatrix} \ti v_{11} & \ti v_{12} d^{2} \\ \ti v_{21} d^{-2}  & \ti v_{22} \end{pmatrix}, &
  \quad z \in \ti\Gamma \setminus  \Gamma^\prime, \\[4mm]
  \begin{pmatrix} \frac{d_-}{d_+} \ti v_{11} & \ti v_{12} d_+ d_- \\
  \ti v_{21} d_+^{-1} d_-^{-1}  & \frac{d_+}{d_-} \ti v_{22} \end{pmatrix}, &
  \quad z \in \Gamma^\prime.
\end{array}\right.
\]

Recall now that
the asymptotic behavior as $z\to 0$ of the solution of the original RHP I -- V  depends on the "slow" variable $\xi=\frac{n}{t}$, and  is determined in essential by the signature table of the real part of the phase function. As it was mentioned in Introduction, we confine ourselves by studying  the solution of  RHP I--V  for $\xi \in [\varepsilon,\,1-\varepsilon]$. The signature table of  $\re \Phi(z,\xi)$ for such $\xi$ is shown in Fig. \ref{fig:Phi-table}.
\begin{figure}[ht]
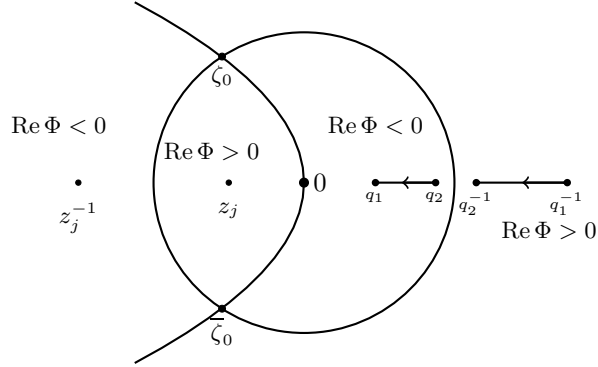

\tikz{
\draw[thick] circle(2);
\draw[thick] plot[domain=-1.5:1.5](-\x*\x,1.6*\x) node[right=2]{};
\filldraw[thick ] (0.95,0)circle(1pt)node[below]{\small$_{q_1}$}--(1.75,0)circle(1pt)node[below]{\small$_{q_2}$};
\draw[thick](1.7,0)--(1.3,0)[->];
\filldraw[thick] (2.29,0)circle(1pt)node[below]{\small$_{q_2^{-1}}$}--(3.5,0)circle(1pt)node[below]{\small$_{q_1^{-1}}$} (0,0)circle(1.5pt);
\draw[thick](3.5,0)--(2.9,0)[->];
%\draw[thick](2,0)arc(0:45:2)[->] node[right]{};
%\draw[thick](-2,0)node[right]{}arc(180:135:2)[<-];
\draw (-2,0.4)node[above=0.3, right=0.4]{\small $\mathrm{Re}\, \Phi > 0$} (0.95,1)node[below]{\small $\mathrm{Re}\, \Phi < 0$}
(2.5,-0.4)node[below right]{\small $\mathrm{Re}\, \Phi > 0$} (-2.5,1)node[below left]{\small $\mathrm{Re}\, \Phi < 0$}  (0,0)node[right]{$0$};
\filldraw[rotate=123] (2,0)circle(1.3pt)node[below]{\small $\zeta_0$};
\filldraw[rotate=-123] (2,0)circle(1.3pt)node[below]{\small $\ol\zeta_0$};
%\draw (-1,0)circle(6pt)node[right=5]{\tiny$\mathbb{T}_j$};
%\draw(-0.8,0)arc(0:35:5pt)[->];
\filldraw  (-1,0)circle(1pt)node[below=4]{\small $z_j$};
%\draw (-3,0)circle(6pt)node[right=5]{\tiny$\mathbb{T}_j^*$};
%\draw(-2.8,0)arc(0:35:5pt)[->];
\filldraw  (-3,0)circle(1pt)node[below=4]{\small $z_j^{-1}$};
}
\caption {Signature table for $\Phi(z)$} \label{fig:Phi-table}
\end{figure}
Evidently, $\re \Phi(z) = 0$ as $z\in\T$.  The other  curve $\re \Phi(z)=0$ intersects $\T$ at the stationary points $\zeta_0$ and $\overline \zeta_0$ of $\Phi$, where $\zeta_0=-\xi+\sqrt{\xi^2 - 1}$.  Our next section consists in a short description of the conjugation/deformation steps performed in \cite{EM} which led to a model RHP for  $\xi \in [\varepsilon,\,1-\varepsilon]$. In \cite{EM} the resonances at points $q_i$ were not admitted. That is why in Section \ref{sec:steps} we make also some additional transformations for the resonant cases.

\section{Reduction to the model problem}\label{sec:steps}

{\it Step 1.} First of all we get rid of singularities of $m$ at the eigenvalues. The  approach to replace the residue conditions by  additional jumps on contours around the eigenvalues was developed in \cite{dkkz}, \cite{KTa}. Let
\beq\nn
P(z)= \prod_{z_j \in (-1,0)} |z_j| \frac{z - z_j^{-1}}{z-z_j}
\eeq
be the Blaschke product corresponding to negative $z_j$
(if any). It satisfies $P(0)>0$, $P(z^{-1})=P^{-1}(z)$. Let $\delta$ be sufficiently small
such that the circles $\mathbb T_j=\{ z : |z - z_j|=\delta\}$ around the eigenvalues
do not intersect and lie away from $\T \cup I$ (the precise value of $\delta$ will be chosen later).
Set
\[
A(z)= \begin{cases}
\begin{pmatrix}1& \frac{z-z_j}{z_j
\gamma_j  \E^{2t\Phi(z_j)} }\\ 0 &1\end{pmatrix}[P(z)]^{-\sigma_3}, & \quad  |z-z_j|< \delta, \quad z_j\in(-1,0),\\
\begin{pmatrix} 1 & 0 \\
\frac{z_j \gamma_j \E^{2t\Phi(z_j)} }{z-z_j} & 1\end{pmatrix}[P(z)]^{-\sigma_3}, & \quad |z-z_j|< \delta, \quad z_j\in(0,1),\\
\si_1 A(z^{-1})\si_1, & \quad |z^{-1}-z_j|< \delta, \quad j=1, \dots, N, \\
[P(z)]^{-\sigma_3}, & \quad \mbox{else.} \\
\end{cases}
\]
We consider the circles $\mathbb T_j$ as contours with counterclockwise orientation.
Denote their images under the map $z\mapsto z^{-1}$ by $\mathbb T_j^*$ and orient them clockwise.

Redefine the solution of RHP I--V by
\beq \nn
m^{\mathrm{ini}}(z)=m(z) A(z),\quad  z\in\C.
\eeq
Then $m^{\mathrm{ini}}(z)$ is a holomorphic vector function in $\C\setminus\{\Gamma\cup \T^\delta\}$,
where $
\T^\delta := \bigcup \limits_j \T_j \cup \T_j^*,
$
and solves the jump problem
$
m_+^{\mathrm{ini}}(z)=m_-^{\mathrm{ini}}(z)v^{\mathrm{ini}}(z),$ for$ z\in\Gamma\cup\T^\delta,
$ where
\[
v^{\mathrm{ini}}(z)= \begin{cases}  v(z),& z\in\Gamma,\\
B(z),& z\in \cup_j \T_j,\\
\sigma_1(B(z^{-1}))^{-1}\sigma_1, &  z\in\cup_j\T_j^*,
\end{cases}
\]
\beq\label{defBII}B(z)= \begin{cases}
\begin{pmatrix}1& \frac{(z-z_j)P^2(z)}{z_j
\gamma_j  \E^{2t\Phi(z_j)} }\\ 0 &1\end{pmatrix}, &  \quad z\in\T_j, \quad z_j\in(-1,0),\\
\begin{pmatrix} 1 & 0 \\
\frac{z_j \gamma_j \E^{2t\Phi(z_j)} }{(z-z_j)P^2(z)} & 1\end{pmatrix}, &  \quad z\in\T_j, \quad z_j\in(0,1). \end{cases}\eeq
Note that
$
\|B(z)-\id\|\leq C\exp\big(- t \inf_j|\re\Phi(z_j)|\big)$ for
$z\in \T^\delta$ when $ 0<\xi<1$.
Here the matrix norm is to be understood as the maximum of the absolute value of its elements.

\vskip 6mm

{\it Step 2.} Set
$z_0=\E^{\I \theta_0}, $ where $\cos \theta_0=1-2\xi,$ $ \theta_0\in(0,\pi),$
and introduce a symmetric contour \beq\nn%\label{siggm}
\Sigma:=\{ z\in\T: \re z\leq \re z_0=\cos\theta_0\},
\eeq
oriented in the same way as $\T$, i.e., from $z_0$ to $\ol z_0$.
Introduce the function
\beq\label{g1}
g(z) = \frac{1}{2} \int_{z_0}^z \sqrt{\Big(1-\frac{1}{sz_0}\Big)\Big(1 - \frac{z_0}{s}\Big)} (1+s) \frac{ds}{s},
\eeq
where the square root in the integrand is defined in $\C\setminus \Sigma$  and $\sqrt{s}>0$ for $s> 0$.

\begin{lemma} \label{lem3.1} \cite{EM} The function $g(z)$ has the following properties:
	
\begin{enumerate} [{\rm (a)}]
\item  Function $g$ has a jump along the arc $\Sigma$ with $g_+(z)=- g_-(z) > 0$ as $z\neq z_0^{\pm 1}$;
\item $\lim \limits_{z\to 0}\Phi(z)-g(z) =  K(\xi)\in\R$, where $
\frac{d}{d\xi} K(\xi) =-\log\xi;
$
\item  $g(\ol{z_0})=0$;
\item  $g(z^{-1}) = - g(z)$ as $z\in\C\setminus\Sigma$.
\end{enumerate}
\end{lemma}

The signature table for $g(z)$ is given in Fig.~\ref{fig:g-table}.
\begin{figure}[ht]
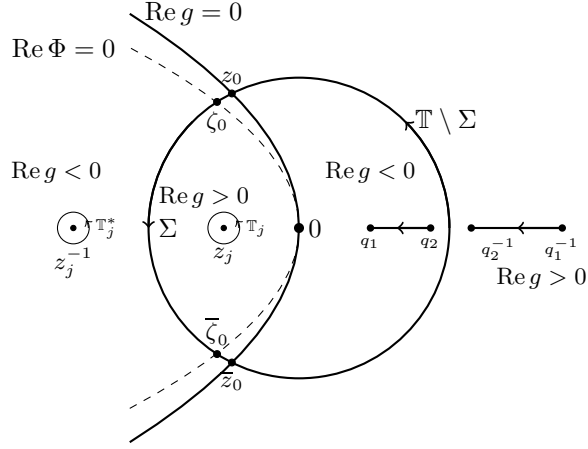

\tikz{
\draw[thick] circle(2);
\draw[dashed] plot[domain=-1.5:1.5](-\x*\x,1.6*\x) node[left=2]{$\mathrm{Re}\, \Phi = 0$};
\draw[thick] plot[domain=-1.5:1.5](-\x*\x,1.9*\x) node[right=2]{$\mathrm{Re}\, g = 0$};
\filldraw[thick ] (0.95,0)circle(1pt)node[below]{\small$_{q_1}$}--(1.75,0)circle(1pt)node[below]{\small$_{q_2}$};
\draw[thick](1.7,0)--(1.3,0)[->];
\filldraw[thick] (2.29,0)circle(1pt)node[below right]{\small$_{q_2^{-1}}$}--(3.5,0)circle(1pt)node[below]{\small$_{q_1^{-1}}$} (0,0)circle(1.5pt);
\draw[thick](3.5,0)--(2.9,0)[->];
\draw[thick](2,0)arc(0:45:2)[->] node[right]{$\mathbb T \setminus \Sigma$};
\draw[thick](-2,0)node[right]{$\Sigma$}arc(180:135:2)[<-];
\draw (-2,0.4)node[above=0.3, right=0.4]{\small $\mathrm{Re}\, g > 0$} (0.95,1)node[below]{\small $\mathrm{Re}\, g < 0$}
(2.5,-0.4)node[below right]{\small $\mathrm{Re}\, g > 0$} (-2.5,1)node[below left]{\small $\mathrm{Re}\, g < 0$}  (0,0)node[right]{$0$};
\filldraw[rotate=116.5] (2,0)node[above]{\small $z_0$}circle(1.3pt);
\filldraw[rotate=-116.5] (2,0)node[below]{\small $\ol z_0$}circle(1.3pt);
\filldraw[rotate=123] (2,0)circle(1.3pt)node[below]{\small $\zeta_0$};
\filldraw[rotate=-123] (2,0)circle(1.3pt)node[above]{\small $\ol\zeta_0$};
\draw (-1,0)circle(6pt)node[right=5]{\tiny$\mathbb{T}_j$};
\draw(-0.8,0)arc(0:35:5pt)[->];
\filldraw  (-1,0)circle(1pt)node[below=4]{\small $z_j$};
\draw (-3,0)circle(6pt)node[right=5]{\tiny$\mathbb{T}_j^*$};
\draw(-2.8,0)arc(0:35:5pt)[->];
\filldraw  (-3,0)circle(1pt)node[below=4]{\small $z_j^{-1}$};
}
\caption {Signature table for g-function} \label{fig:g-table}
\end{figure}
From Lemma \ref{lem3.1} and from the oddness of the phase function $\Phi(z)$ it follows that
the function  $d(z): =\E^{t(\Phi(z) - g(z))}$. satisfies our symmetry constraints.
Set
\beq \nn
m^{(1)}(z)=m^{\mathrm{ini}}(z)[ d(z)]^{-\si_3},
\eeq
then the vector function $m^{(1)}(z)$ solves  jump problem $m_+^{(1)}(z)=m_-^{(1)}(z)v^{(1)}(z)$ with
\[
v^{(1)}(z)=\left\{
\begin{array}{ll}
\begin{pmatrix}
0 & - \ol{\mathcal R(z)} \E^{- 2 t g(z)} \\
\mathcal R(z) \E^{2 t g(z)} & 1
\end{pmatrix}, & z \in \T \setminus \Sigma,\\[3mm]
\begin{pmatrix}
0 & - \ol{\mathcal R(z)} \\
\mathcal R(z)  & \E^{-2 t g_+(z)}
\end{pmatrix}, & z \in \Sigma,\\[3mm]
E(z),  & z \in \Xi,
\end{array}\right.
\]
where
$
\mathcal R(z) :=R(z)P^{-2}(z),$  $\Xi:=I\cup I^*\cup \T^\delta,$ and
\begin{align}\label{defB11}
 E(z)&:= \begin{cases}
\begin{pmatrix}
1 & 0 \\
P^{-2}(z)\chi(z) \E^{2t g (z)} & 1
\end{pmatrix}, & z \in I, \\[3mm]
\sigma_1(E(z^{-1}))^{-1}\sigma_1, & z\in I^*, \\
[d(z)]^{\si_3}B(z)[d(z)]^{-\si_3}, &  z\in \T^\delta.
\end{cases}
\end{align}
As it is shown in \cite{EM}, for sufficiently small $\delta>0$ in the case if the points $q_1$ and $q_2$ are nonresonant, the following estimate is valid:
\beq\label{estE_1}
\|E(z) - \id\|_{L^\infty(\Xi)}\leq C\E^{-\frac{t J(\delta)}{2}}, \quad J(\delta)>0.
\eeq

\vskip 5mm

{\it Step 3}:  Here we perform a standard lenses mechanism related to the upper-lower factorization of the jump matrix (\cite{dz, KTb}).  Till this step we did not use the decaying condition \eqref{decay}. To proceed further one has to specify the value of $\nu$ there. Indeed, the lenses mechanism requires a holomorphic continuation of the right reflection coefficient in a small vicinity of the arc $\T\setminus \Sigma$. Denote $\hat \psi(z,n):=\psi(z,n,t)|_{t=0}$, $\hat \psi_1(z,n):=\psi(z,n,t)|_{t=0}$. Since \[R(z)=-\frac{\langle\hat \psi_1(z),\,\overline{\hat\psi(z) }\rangle}{\langle\hat \psi_1(z),\,\overline{\hat\psi(z) }\rangle},\]
and $\overline R(z)=R(z^{-1})=R^{-1}(z)$ as $z\in\T$, we need an analytical continuation $\breve \psi(z,n)$ of
 the function $\overline{\hat \psi(z,n)}, $ which is defined initially on $z\in\T$. Evidently, this continuation
can be represented via the transformation operator
\[\breve\psi(z,n)=\sum_{m=n}^\infty K(n,m)z^{-m}, \quad z\in\T,\]
where for $ m>n>0$ the following estimate holds:
\beq\nn
|K(n,m)|\leq  C(n)\sum_{l=[\frac{m+n}{2}]}^\infty\{ |a(n,0)-\frac 1 2| + |b(n,0)|\},\quad 0<C(n)<C.
\eeq
Thus for arbitrary  $0<\nu$ in \eqref{decay} the right reflection coefficient can be continued analytically in the domain $\mathfrak P_\nu:=\{z:\ \E^{-\nu}<|z|<1,\ z\notin I\}$, with possible poles at points of the discrete spectrum which find themselves inside this ring with a cut. The lenses mechanism which we apply  holds for an arbitrary small $\nu>0$ if there are no resonances at points $q_1$ and $q_2$, for $\nu> - \log q_2$ if the point $q_1$ is not resonant, and for $\nu>-\log q_1$ if there is a resonance at $q_1$.
The case of the nonresonant points $q_1$ and $q_2$ is described in \cite{EM}.  Assume that the  point $q_2$ is the resonant one and $\nu>-\log q_1$.
Let $\mathcal C$ be a contour close to the arc $\T\setminus \Sigma$ with endpoints $z_0$ and $\ol z_0$ and clockwise orientation as depicted on Figure \ref{fig:ContDeform}. We presume that there are no points of the discrete spectrum between this contour and the arc  $\T\setminus \Sigma$.   Choose another contour $\mathcal C_I$ close the interval $I$ such that $\mathcal C_I\subset \mathfrak P_\nu$ and there are no points of the discrete spectrum inside  $\mathcal C_I$. Orient this contour counterclockwise.

 Let now $\mathcal C^*$ (resp., $\mathcal C^*_I$) be the image of $\mathcal C$ (resp., $\mathcal C_I$ ) under the map $z\mapsto z^{-1}$, oriented clockwise (resp., counterclockwise)  as well.
\begin{figure}[ht]{
\tikz{
\draw[dashed]circle(2);
\draw[rotate=120] (2,0) coordinate (a);
\draw[rotate=-120] (2,0) coordinate (b);
\draw[thick] (a)arc(120:240:2);
\draw[rotate=3](2,0) coordinate (c);
%\draw (c)--(0,0);
\draw[rotate=-3](2,0) coordinate (d);
%\draw (d)--(0,0);
%\draw (c)..controls (0,0)..(d);
\draw (1.8,0)arc(0:117:1.8)--(a)[red][thick];
\draw (2.2,0)arc(0:117:2.2)coordinate(c)--(a)[blue][thick];
\draw[thick] (0.7,0)ellipse(14pt and 5 pt)[red];
\draw[thick] (1.18,0)arc(0:90:14pt and 5 pt)[->][red];
\draw (1.8,0)arc(0:-117:1.8)--(b)[red][thick];
\draw (2.2,0)arc(0:-117:2.2)--(b)[blue][thick];
\draw[thick] (4,0)ellipse(28pt and 7  pt)[blue][->];
\draw[thick] (4.98,0)arc(0:90:28pt and 7  pt)[->][blue];
\draw[thick](a)arc(120:165:2)[->]node[right]{$\Sigma$};
\draw[thick](0,1.8)arc(90:55:1.8)[->]node[below left]{$C$}[red] ;
\draw[thick](0,2.2)arc(90:45:2.2)[->]node[above right]{$C^*$}[blue] ;
\filldraw[thick] (3.2,0)circle(0.8pt)node[below =2]{\tiny${q_2^{-1}}$}--(4.8,0)circle(0.8pt)node[below=2]{\tiny${q_1^{-1}}$} (1.55,0)circle(0.8pt)node[below=2]{\tiny$z_j$}
(2.7,0)circle(0.8pt)node[below=2]{\tiny$z^{-1}_j$}
 (0,0)circle(1.4pt) node[left] (0,0) {\small ${0}$}
 (a)circle(1pt)node[below]{\tiny $z_0$}
 (b)circle(1pt)node[above]{\tiny $\overline{z_0}$};
\filldraw[thick] (0.35,0)circle(0.8pt)node[below=2]{\tiny${q_1}$}--(1.1,0)circle(0.8pt)node[below=2]{\tiny$_{q_2}$};
\draw[thick](4.8,0)--(4,0) [->]node[above=6]{\small$C_I^*$};
\draw[thick](1.1,0)--(0.65,0)[->]node[above=4]{\small$C_I$};
\draw (1.55,0)circle(5pt)node[above=2]{\tiny$\mathbb{T}_j$};
\draw(1.72,0)arc(0:35:5pt)[->];
 \draw(2.7,0)circle(5pt)node[above=2]{\tiny$ \mathbb{T}_j^*$};
 \draw(2.87,0)arc(0:35:5pt)[->];
%\draw[rotate around={30:(b)}] (0,0)--(b)[dashed]1;
%\draw[thick] (5,0)arc(-180:-90:28.5pt and 8  pt)[red];
}
\label{fig:ContDeform}
}
\end{figure}
Let $\Omega$ consist of two domains: the first one is located between $\mathcal C$ and $\T\setminus\Sigma$, and the second one is a region inside $\mathcal C_I$ without points of $I$.  Denote the symmetric domains by  $\Omega^*$ as in Fig.~\ref{fig:ContDeform}, and set
\beq \nn
V(z):= \begin{cases}
	\begin{pmatrix}
	1 & 0 \\
	- \mathcal R(z) \E^{2 t g(z)} & 1
	\end{pmatrix}, & z \in \Omega,\\
	\sigma_1(V(z))^{-1}\sigma_1, &  z\in\Omega^*
	%,\\ \id, & z\in\C\setminus(\ol{\Omega\cup\Omega^*})
	.
	\end{cases}
	\eeq
Set
\beq\nn
m^{(2)}(z)=m^{(1)}(z) V(z),\ \ z\in\Omega\cup\Omega^*; \ \ m^{(2)}(z)=m^{(1)}(z), \ \  z\in\C\setminus(\ol{\Omega\cup\Omega^*}).
\eeq
The new vector function has not  jump on $\T\setminus\Sigma$. Instead it has additional jumps on the new contours and satisfies  $m_+^{(2)}(z)=m_-^{(2)}(z)v^{(2)}(z)$, where (cf. \eqref{defB11},\eqref{defBII})
\beq\nn%\label{v2}
v^{(2)}(z)=\begin{cases} v^{(1)}(z), & z \in \Sigma \cup \T^\delta,\\
V(z), & z\in \mathcal C\cup\mathcal C_I,\\
V_-^{-1}(z)E(z)V_+(z), & z\in I,\\
\sigma_1[v^{(2)}(z^{-1})]^{-1}\sigma_1, & z\in \mathcal C^*\cup\mathcal C_I^*\cup I^*.
\end{cases}
\eeq
Note that this step again preserves both the normalization and the symmetry conditions.
Next, we observe that
\[
V_-^{-1}(z)E(z)V_+(z)= \begin{pmatrix} 1 & 0\\ P^{-2}(z)\left[R_-(z) + \chi(z) - R_+(z)\right]\E^{2 t g(z)} & 1 \end{pmatrix},\quad z \in I.
\]
Recall now that that
\[
R_- = -\frac{\big<\ol \psi_1, \breve \psi \big>}{\big<\ol \psi_1,  \psi \big>},\,  \chi = -\frac{\big<\breve \psi, \psi \big>}{\big< \ol\psi_1,  \psi \big>} \frac{\big<\ol \psi_1, \psi_1\big>}{\big< \psi_1, \psi \big>},
\, R_+ = -\frac{\big< \psi_1,  \breve \psi \big>}{\big<\psi_1 , \psi \big>}.
\]
As the Wronskian of any two Jost solutions is an anti-symmetric operation, by the Plucker identity
\[
\big<a,b\big>\big<c,d\big> + \big<a,c\big>\big<d,b\big> + \big<a,d\big>\big<b,c\big> \equiv 0,
\]
setting $a = \breve \psi ,\, b = \psi,\, c = \ol \psi_1,\, d = \psi_1$, we have $R_-(z) + \chi(z) - R_+(z)\equiv 0\,$ and  \beq\nn V_-^{-1}(z)E(z)V_+(z) = \id,\ \ \mbox{for}\ \ z\in I.\eeq
Thus, this step allows us to get rid of a possible   singularity of the jump matrix in the resonant point under assumption that the initial data tend to the right background constant exponentially fast with an exponent which depends on a location of the resonant point. Since this singularity is not $L_2$ this transformation will allow us further to stay in the frameworks of the standard schemes of investigation of the Cauchy-type integrals.

Note that in the case of the absence of resonances in $q_1$ and $q_2$ we do not need to introduce contours $\mathcal C_I$ and $\mathcal C_I^*$ and redefine $m^{(1)}$ inside them, because estimate \eqref{estE_1} is already suitable for further considerations. To put all cases (resonant and nonresonant) in a single scheme we can either set $v^{(2)}(z)=\id$ on $\mathcal C_I\cup\mathcal C_I^*$. Anyway, if we denote
\beq\label{hatxi}\breve\Xi:=I\cup I^*\cup \T^\delta\cup\mathcal C_I\cup \mathcal C_I^*,\eeq
we observe that $v^{(2)}(z^{-1})=\sigma_1[v^{(2)}(z)]^{-1}\sigma_1$ for $z\in\breve\Xi$ and
\beq\label{brevexi}
\|v^{(2)}(z) -\id\|_{L^\infty(\breve\Xi)}\leq C\E^{-t \mu},\quad C=C(\nu, \epsilon)>0, \ \ \mu=\mu(\nu, \varepsilon)>0.\eeq
\vskip 5mm

{\it Step 4}:  For $z\in\C\setminus \Sigma$ introduce the function \[q(z,z_0)=\sqrt{(z_0 -z)(z_0z-1)}.\] Let $[q(z,z_0)]_+$ be its values from the positive side of $\Sigma$. Set
\beq\label{13}\breve d(z)=\exp \left( \frac{q(z, z_0)}{2 \pi\I}
\int_{z_0}^{\ol{z_0}} \frac{ \log\frac{ \mathcal R(s)}{\mathcal R(-1)}\,ds}{[q(z,z_0)]_+(s-z)}\right).\eeq
As is shown in \cite{EM}, this function uniquely solves the
following scalar conjugation problem:  find a holomorphic function $\tilde d(z)$ on $\C \setminus \Sigma$, such that
\begin{align} \label{scp}
& \breve d_+(z) \breve d_-(z) = \mathcal R(z) \mathcal R^{-1}(-1), \quad z \in \Sigma,  \\
& \mbox{(i)}\ \breve d(z^{-1}) = \breve d^{-1}(z), \quad z \in \C \setminus \Sigma; \quad \mbox{(ii)}\ \breve d(\infty)>0.\label{symmd}
\end{align}
Define now $m^{(3)}(z)=m^{(2)}(z)[\breve d(z)]^{-\sigma_3},$ $z\in \C$. From the previous considerations we conclude that $m^{(3)}(z)$ is the unique solution of the following  RHP, which is in fact RHP--{\it equiv}:  find a holomorphic
vector function $\breve m(z)$ in $\C\setminus(\Sigma\cup \mathcal C\cup\mathcal C^*\cup \breve\Xi)$, which is continuous up to
the boundary and has the following properties:
\begin{itemize} \item It solves the jump problem $\breve m_+(z)=\breve m_-(z) \breve v(z)$ with
\beq \label{v3}
\breve v(z)=\left\{\begin{array}{ll}
\begin{pmatrix}
0 & - \mathcal R(-1) \\
\mathcal R(-1) & \frac{\breve d_+(z)}{\breve d_-(z)} \E^{-2tg_+(z)}
\end{pmatrix}, & z \in \Sigma,\\[4mm]
\begin{pmatrix}
1 & 0 \\
-\breve d^{-2}(z) \mathcal R(z) \E^{2 t g(z)} & 1
\end{pmatrix}, & z \in \mathcal C,\\[3mm]
\begin{pmatrix}
1 & \breve d^2(z) \ol{\mathcal R(z)} \E^{- 2 t g(z)} \\
0 & 1
\end{pmatrix}, & z \in \mathcal C^*,\\[3mm]
[\breve d(z)]^{\sigma_3}v^{(2)}(z)[\breve d(z)]^{-\sigma_3},  & z \in \breve\Xi,
\end{array}\right.
\eeq
From \eqref{brevexi} it follows that
\beq\label{brevexi}
\|\breve v(z) -\id\|_{L^\infty(\breve\Xi)}\leq C\E^{-t \mu},\quad C=C(\nu, \epsilon)>0, \ \ \mu=\mu(\nu, \varepsilon)>0.\eeq
\item $\breve m(z^{-1})=\breve m(z)\sigma_1$ and $\breve m_1(0)\,\breve m_2(0)=1$; moreover,  $\breve v(z)$ has the symmetry property \eqref{matsym}.
\item
For small $z$ the original vector function $m$ in \eqref{defm} and $\breve m(z)$ are connected by the following transformation
\beq \label{conn}
\breve m(z)=m(z) [q(z)]^{-\sigma_3},\quad q(z)=\breve d(z) P(z)
\E^{t(\Phi(z) - g(z))}.
\eeq
\end{itemize}
Denote
\beq\label{vmod}
v^{\mathrm{mod}}(z)=\begin{pmatrix}
0 & -\mathcal R(-1) \\
\mathcal R(-1) & 0
\end{pmatrix},  z \in \Sigma.
\eeq
Let $\mathcal B$ and $\mathcal B^*$ be small vicinities of points $z_0$ and $\overline z_0$ respectively. We suppose that these vicinities do not have joint points on their boundaries. They are symmetric with respect to the map $z\mapsto z^{-1}$ and are located on  positive distances bigger than $\frac{\im z_0}{2}$ from points $\pm 1$.
Put $\breve\Gamma:=\Sigma \cup \mathcal C\cup\mathcal C^* \cup \breve\Xi$, where $\breve \Xi$ is defined by \eqref{hatxi}.  Taking into account the signature table for $\re g(z)$ for $\xi\in [\varepsilon, 1-\varepsilon$, we observe that
\beq\label{v3vmodnorm}
\|\breve v(z)-v^{\mathrm{mod}}(z)\|_{L^\infty\left(\breve\Gamma\setminus(\breve\Gamma\cap({\mathcal B}\cup {\mathcal B}^*))\right)}\leq C\E^{-t U},\quad U=U(\varepsilon)>0.
\eeq

Therefore one can assume that the solution of the  RHP--{\it equiv}  \eqref{v3} can be well approximated in $\C\setminus(\mathcal B\cup\mathcal B^*)$ by a solution of the following model RHP:  find a holomorphic vector function in $\C \setminus \Sigma$ satisfying the jump condition
\beq \label{modRHP}
m_+^{\mathrm{mod}}(z)= m_-^{\mathrm{mod}}(z) v^{\mathrm{mod}}(z), \quad z \in \Sigma,
\eeq
the symmetry condition $m^{\mathrm{mod}}(z^{-1})=m^{\mathrm{mod}}(z)\si_1$, and the
normalization condition $m_1^{\mathrm{mod}}(0)>0$,   $m_1^{\mathrm{mod}}(0)m_2^{\mathrm{mod}}(0) =1$.

The solution of this problem is unique (see lemma \ref{ModProbUniq} in appendix). In \cite{EM} it was proved that
\beq\label{mmod} m^{\mathrm{mod}}(z)=(\alpha, \alpha) M^{\mathrm{mod}}(z), \quad \mbox{where}\ \
\big(\sin\tfrac{\theta_0}{2}\big)^{-1/2};\eeq
\beq \label{Mmod}
M^{\mathrm{mod}}(z)=\begin{pmatrix} \frac{\beta(z) + \beta^{-1}(z)}{2} & \frac{\beta(z) - \beta^{-1}(z)}{2 \I} \\[2mm]
- \frac{\beta(z) - \beta^{-1}(z)}{2 \I} & \frac{\beta(z) + \beta^{-1}(z)}{2} \end{pmatrix};
\eeq
\beq\label{beta}
 \beta(z)=\left(\frac{z_0 z - 1}{z_0 - z}\right)^{\pm1/4}, \ \mbox{for} \ \mathcal R(1)=\mp 1.
\eeq
Here the branch of the fourth root is chosen with the cut along the negative half axis and $1^{1/4}=1$. The case $\mathcal R(1)=-1$ corresponds to the nonresonant case at $1$, while the case $\mathcal R(1)=1$ is related to the resonant case. Note that
 $M^{\mathrm{mod}}(z)=M^{\mathrm{mod}}(z,n,t)$ solves the following model matrix RHP:
find a holomorphic matrix function $M^{\mathrm{mod}}(z)$ on $\C \setminus \Sigma$ satisfying the following jump and symmetry conditions,
\beq\nn%\label{matrixmod}
M_+^{\mathrm{mod}}(z)= M_-^{\mathrm{mod}}(z)
 v^{\mathrm{mod}}(z), \quad z \in \Sigma; \qquad M^{\mathrm{mod}}(z^{-1})=\sigma_1 M^{\mathrm{mod}}(z)\sigma_1.
\eeq
The solution of this matrix  RHP is not unique.

\section{The parametrix problem}\label{parametrix}

\begin{align} \label{parsets}
\breve\Gamma&=\Sigma \cup \mathcal C\cup\mathcal C^* \cup \breve\Xi\cup\partial \mathcal B\cup\partial \mathcal B^*,\quad
\Sigma_\mathcal B=\breve\Gamma\cap \mathcal B,\\ \nn
\Sigma_1&=\Sigma \cap \mathcal B,\quad \Sigma_2=\mathcal C\cap\mathcal B,\quad  \Sigma_3=\mathcal C^*\cap\mathcal B. \nn
\end{align}

Properties of the function $g$ don't allow us to use the same arguments considering close to $z_0, \ol z_0$ points of the contour $\ti \Sigma$ , so an additional consideration of this part is required.  Technically this problem is similar to the one considered in \cite{dkmv}.

Here we give the lemma proved in \cite{EM}, which allows us to replace the jump matrix $v^{(3)}(z)$ inside $\mathcal B$ by another approximately close matrix.
\begin{lemma}
The function $\tilde d(z)$ satisfying \eqref{scp}--\eqref{symmd} has the following asymptotic behavior
in a vicinity of $z_0$,
\beq \label{dlim}
\tilde d^{-2}(z) \mathcal R(z)= \mathcal R(-1) +O(\sqrt{z-z_0}), \ z \notin \Sigma,
\quad \frac{ \tilde d_+(z)}{ \tilde d_-(z)} = 1+O(\sqrt{z-z_0}), \ z \in \Sigma.
\eeq
\end{lemma}

\begin{proof}
We will use the representation \eqref{13}. To simplify notation set
 \beq\nn%\label{temp}
r(s) = \log( \mathcal R(s)\mathcal R^{-1}(-1)),\quad q(s,z_0)=\sqrt{(z_0-z)(z_0z - 1)}.
\eeq
Then
\[
\log \tilde d(z)= \frac{q(z,z_0)}{2\pi\I} \int_{z_0}^{\ol z_0} \frac{r(s)ds}{[q(s, z_0)]_+ (s-z)}=J_1(z) + J_2(z),
\]
where
\begin{align*}
J_1(z)&= \frac{q(z,z_0)}{2\pi\I} \int_{z_0}^{\ol z_0} \frac{(r(s)-r(z))ds}{[q(s, z_0)]_+ (s-z)}, \\
J_2(z)&=r(z)\frac{q(z,z_0)}{2\pi\I}\int_{z_0}^{\ol z_0} \frac{ds}{[q(s, z_0)]_+ (s-z)}.
\end{align*}
Since $r(s)-r(z)\sim (s-z)$, the integral in $J_1(z)$ is H\"older continuous in a vicinity
of $z_0$. Therefore,
\beq
\begin{aligned} \label{asymp9}
J_1(z) & =I(z_0)\sqrt{z-z_0} (1 +o(1)), \\
I(z_0)&=\frac{\sqrt{1-z_0^2}}{2\pi\I} \int_{z_0}^{\ol z_0} \frac{(\log\mathcal R(s)-\log\mathcal R(z_0))\,ds}{[\sqrt{(z_0 -s)(z_0s-1)}]_+ (s-z_0)}.
\end{aligned}
\eeq
On the other hand, since
\[
\frac{1}{2[q(z,z_0)]_+}=\frac{1}{2[q(z,z_0)]_-} + \frac{1}{[q(z,z_0)]_+},\quad z\in\Sigma,
\]
and $(q(z,z_0))^{-1}\to 0$ as $z\to\infty$, we have
\[
\frac{1}{2q(z,z_0)}=\frac{1}{2\pi\I}\int_{z_0}^{\ol z_0} \frac{ds}{[q(s, z_0)]_+ (s-z)},
\]
and $J_2(z)=\frac{r(z_0)}{2}(1 +O(z-z_0))$. Therefore,
\beq\nn%\label{asymp17}
\log\tilde d(z)=\frac{1}{2}\log\frac{\mathcal R(z_0)}{\mathcal R(-1)} + I(z_0)\sqrt{z-z_0} + o(\sqrt{z-z_0}),
\eeq
and \eqref{dlim} follows from \eqref{asymp9} in a straightforward manner.
\end{proof}
Using this lemma we have
\beq\nn
v^{\mathrm{par}}(z) :=\E^{-t g_-(z)\sigma_3}\, S\, \E^{t g_+(z)\sigma_3},
\eeq
where
\beq\nn S=\begin{cases} S_1:=\begin{pmatrix}  0&-\mathcal R(-1)\\ \mathcal R(-1)& 1\end{pmatrix}, & z\in\Sigma_1,\\[3mm]
 S_2:= \begin{pmatrix}  1&0\\ -\mathcal R(-1)& 1\end{pmatrix} , & z\in\Sigma_2,\\[3mm]
 S_3:= \begin{pmatrix}  1&\mathcal R(-1)\\0& 1\end{pmatrix}, & z\in\Sigma_3,\end{cases}
\eeq
where $\Sigma_i$ defined by \eqref{parsets} and oriented outwards.

Note, that if $M^{\mathrm{par}}$ solves \eqref{parini}, then the matrix function
\[
M(z)=M^{\mathrm{par}}(z)\E^{-t g(z)\sigma_3}
\]
solves the constant jump problem
\[
M_+(z) := M_-(z) S
\]
with the normalization $M\sim M^{\mathrm{mod}}\E^{-t g(z)\sigma_3}$ on $\partial\mathcal B$.

To simplify our considerations we will next use a change of coordinates
\beq\label{change}
w(z)=\left(\frac{3 t g(z)}{2}\right)^{2/3}.
\eeq

\begin{lemma}\label{gexp}
In a vicinity of $z_0$,
\beq\nn
g(z)= C(\theta_0) (z-z_0)^{\frac{3}{2}} \E^{\I (\frac{\pi}{4}- \frac{3\theta_0}{2})}(1 + o(1)),
\eeq
\[
C(\theta_0)=\frac{\sqrt{2}}{3}(\sin{\theta_0})^{\frac{1}{2}}\cos{\frac{\theta_0}{2}}>0.
\]
\end{lemma}
\begin{proof}
The function $g(z)$ is given by \eqref{g1}. Then it's derivative is
\[
\frac{d}{dz}g(z) = \frac{1}{2} \sqrt{\Big(1-\frac{1}{zz_0}\Big)\Big(1 - \frac{z_0}{z}\Big)} \frac{1+z}{z}
	=  \sqrt{z-z_0}f(z),
\]
where $f(z) = \frac{1}{2} \frac{1+z}{z\sqrt{z}}\sqrt{1-\frac{1}{zz_0}}$ -- holomorphic in a vicinity of $z_0$, and
\[
f(z_0) = 2\sqrt{2}\E^{\I(\frac{\pi}{4} - \frac{3\theta_0}{2})}\sqrt{\sin{\theta_0}}\cos{\frac{\theta_0}{2}}.
\]
\noprint {This allows us
to get an expansion of $g^{\prime}(z)$.
\begin{align*}
f(z_0) &=\frac{1+z_0}{2z_0}\sqrt{\frac{1}{z_0}}\sqrt{1-\frac{1}{z^2_0}} = [z_0 = \E^{\I \theta_0}]=
    \sqrt{\frac{\E^{2\I \theta_0}-1}{\E^{2\I \theta_0}}}\sqrt{\frac{1}{\E^{\I \theta_0}}}\frac{1+\E^{\I \theta_0}}{2\E^{\I \theta_0}}\\
    & = \frac{\E^{\I \frac{\theta_0}{2}}+\E^{\I \frac{\theta_0}{2}}}{2\E^{\I \frac{\theta_0}{2}}}\Big(\frac{\E^{\I \theta_0}-\E^{-\I \theta_0}}{\E^{2\I \theta_0}}\Big)^\frac{1}{2} = \frac{\sqrt{2\I\sin{\theta_0}}}{\E^{\I \theta_0}}\frac{2\cos{\frac{\theta_0}{2}}}{\E^{\I \frac{\theta_0}{2}}}, \quad \sqrt{\I} = \E^{\I\frac{\pi}{4}}\\
    &
 \end{align*}

In a vicinity of $z_0$
\[
    f(z) = f(z_0) + O\big((z-z_0)\big).
\]
}
So
\[
g^{\prime}(z) = 2^\frac{3}{2}(\sin{\theta_0})^{\frac{1}{2}}\cos{\frac{\theta_0}{2}}\E^{\I(\frac{\pi}{4} - \frac{3\theta_0}{2})} \sqrt{z-z_0}\left(1+O\big((z-z_0)\big)\right),
\]
and
\[
g(z) = C(\theta_0)\E^{\I(\frac{\pi}{4} - \frac{3\theta_0}{2})} (z-z_0)^\frac{3}{2}(1 + o(1)), \quad z \to z_0.
\]
\noprint{
\begin{align*}
g(z)& = \int \limits_{z_0}^z g^{\prime}(s) ds = \frac{\sqrt{2}}{3}(\sin{\theta_0})^{\frac{1}{2}}\cos{\frac{\theta_0}{2}}\E^{\I(\frac{\pi}{4} - \frac{3\theta_0}{2})} (z-z_0)^\frac{3}{2}(1 + o(1))\\
& = C(\theta_0)\E^{\I(\frac{\pi}{4} - \frac{3\theta_0}{2})} (z-z_0)^\frac{3}{2}(1 + o(1)).
\end{align*}
}
Since $\quad \theta_0 \in \left(\frac{\pi}{2},\pi\right), \quad$ $\sin{\theta_0},\ \cos{\frac{\theta_0}{2}} > 0 \quad$ and  $\quad C(\theta_0)>0$.
\end{proof}
Then we have
\beq\label{wz0}
w(z) = t^\frac{2}{3}\left(\frac{3}{2}C(\theta_0)\right)^\frac{2}{3}\E^{\I(\frac{\pi}{6} - \theta_0)}(z-z_0)(1+\disp o(1)),\quad \mbox{as} \quad z \to z_0.
\eeq

Choose the set $\mathcal B$ as the preimage under the map $z\mapsto w$ of the circle $\quad \mathbb D_\rho = \{w \in \C: |w|<t^{2/3} \left(\frac{3}{2}C(\theta_0)\right)^{2/3} \rho, \quad\rho<\im z_0/2 \}$, centered at $w=0$. Seeing we have some flexibility of choosing
the contours $\mathcal C$ and $\mathcal C^*$ we can set them such that $\Sigma_1\cup\Sigma_2\cup \Sigma_3$ are mapped onto the straight lines $\left(\Gamma_1\cup\Gamma_2\cup\Gamma_3\right)\cap \mathbb D_\rho$(see Fig. \ref{fig:ZtoW}), where
\[
\Gamma_2=\{w\in\C:\,\arg w=\frac{4\pi}{3}\},\  \Gamma_3=\{w\in\C:\,\arg w=\frac{2\pi}{3}\},\  \Gamma_1=[0,\,+\infty).
\]
\begin{figure}[ht]
\tikz[rotate=30]{
\clip[draw] ellipse(2.4 and 1.8);
%\draw[thick] (0,-4)circle(4);
\draw[thick] (0,0)arc(90:125:4);
\draw[thick] (0,0)arc(90:105:4)[->]node[above left]{\small $\Sigma_1$};
\draw[dashed] (0,0)arc(90:55:4);
\draw[thick] ellipse(2.39 and 1.79);
\filldraw circle(2pt) node[above left] (0,0) {\small $z_0$};
\draw[thick](0,0)arc(145: 120: 3 and 2)[->]node[above left]{\small $\Sigma_3$};
\draw[thick] (0,0)arc(205: 235: 3 and 2)[->]node[below]{\small $\Sigma_2$};
\draw[thick](0,0)arc(145: 90: 3 and 2);
\draw[thick] (0,0)arc(205: 300: 3 and 2);
\draw[rotate=-140] (1.2,0) node[below]{\large $\mathcal B$};
}
\tikz{
\draw[thick] (0,1.8)arc(115:105: 3 and 2.5) node[above right]{$w = w(z)$} (0,1.8)arc(115:65: 3 and 2.5) [->];
\node[above] (0,0) {\small $ $};
}
\tikz{
\draw[thick] circle(1.99);
\filldraw circle(2pt) node[above right]{\small ${0}$};
\draw[thick] (0,0)--(1.99,0) (0,0)--(1.4,0)[->] node[above]{\small $\Gamma_1$};
\draw[thick][rotate=120] (0,0)--(1.99,0) (0,0)--(1.4,0)[->]node[left]{\small $\Gamma_2$};
\draw[thick][rotate=240] (0,0)--(1.99,0) (0,0)--(1.4,0)[->]node[right]{\small $\Gamma_3$};
\draw[rotate=-55] (1.4,0) node[below]{\large $\mathbb D_\rho$};
}
\caption {The map $z \to w$}\label{fig:ZtoW}
\end{figure}
\begin{lemma}
$w(z)$ maps $\Sigma_1$ into $[0,A],\,A>0$.
\end{lemma}
\begin{proof}

Since $\arg (z-z_0) = \theta_0 - \frac{3\pi}{2},$ then
\[\E^{\I(\frac{\pi}{4} - \frac{3\theta_0}{2})} (z-z_0)^\frac{3}{2} = \E^{-2\pi\I}|z-z_0| = |z-z_0|.\]
\noprint{\[
\left(|z-z_0| \E^{\I(\frac{\pi}{6} - \theta_0+\theta_0 - \frac{3\pi}{2})}\right)^{3/2} = |z-z_0| \E^{\I\left(\frac{\pi}{6} -\frac{3\pi}{2}\right)} =|z-z_0| \E^{-\I\frac{4\pi}{3}}
\]
and
\[
\left(\E^{\I(\frac{\pi}{4} - \frac{3\theta_0}{2})} (z-z_0)^\frac{3}{2}\right)^{2/3} = |z-z_0| \E^{-8\pi\I} = |z-z_0|.
\]
}
\end{proof}
Consider first the generic nonresonant case where
\beq\nn
S_1=\begin{pmatrix} 0& 1\\-1 & 1\end{pmatrix},\ \ S_2=\begin{pmatrix} 1& 0\\1 & 1\end{pmatrix},\ \ S_3=\begin{pmatrix} 1& -1\\0 & 1\end{pmatrix},
\eeq
and the function $\beta(z)$ is locally given by \eqref{beta}
\[
\beta(z)= w^{-1/4} \gamma(w), \qquad w\in\mathbb D_\rho,
\]
where $\gamma$ is holomorphic, and
\beq \label{DefGamma}
\gamma(w) = \beta(z(w))w^{1/4} = 2^{\frac{1}{6}}\big(\sin\theta_0\big)^{1/3}\big(\cos\tfrac{\theta_0}{2}\big)^{1/6}t^{1/6} \E^{-i\frac{\pi}{12}}\left(1+O\big(t^{-2/3}w\big)\right).
\eeq
Combining the fact that $\gamma(w)$ holomorphic with \eqref{beta} we can represent matrix $M^{\mathrm{mod}}(z)$ as
\beq\nn
 M^{\mathrm{mod}}(z)= M_0 [\gamma(w)]^{\sigma_3} w^{-\frac{\sigma_3}{4}} M_0^{-1}.
\eeq
Where $M_0 = \begin{pmatrix} 1 & 1 \\ \I & -\I \end{pmatrix}$. Since $\mathbb D_\rho$ grows as $t\to\infty$, if we find a matrix solution $\mathcal A(w)$ of the jump problem
 \beq\label{jump99}\mathcal A_+ = \mathcal A_-S_j  \quad \mbox{on}\quad \Gamma_j,
\eeq
 with the normalization condition
\beq\label{normairy}
\mathcal A(w) = w^{-\frac{\sigma_3}{4}}(M_0^{-1} +O(w^{-3/2}))\E^{-\frac{2}{3} w^{3/2}\sigma_3}, \quad \text{as}\quad w\to\infty,
\eeq
in any direction with respect to $w$ (here the term $O(w^{-3/2})$ has been written a-posteriori), then the matrix
\begin{align}\label{mpar}
M^{\mathrm{par}}(z) &= M_0 [\gamma(w)]^{\sigma_3} \mathcal A(w) \E^{\frac{2}{3} w^{3/2}\sigma_3}\\
&= M^{\mathrm{mod}}(z) M_0 \left[\tfrac{3 t g(z)}{2}\right]^{\sigma_3/6} \mathcal A\left(\big(\tfrac{3t g(z)}{2}\big)^{2/3}\right) \E^{t g(z) \sigma_3} \nn
\end{align}
will satisfy the condition
\beq\label{mparmod}
M^{\mathrm{par}}(z)= M^{\mathrm{mod}}(z) \big(\id +O(\rho^{-3/2}t^{-1})\big), \quad \text{as}\quad t\to\infty, \quad z\in \partial \mathcal B.
\eeq
So  the parametrix problem \eqref{parini} was reduced to the problem \eqref{jump99}, \eqref{normairy}. The solution of the last one can be given in terms of Airy functions. To this end set
\begin{align*}
y_1(w)&:=\Ai(w)=\frac{1}{2\pi \I}\int \limits_{-\I \infty}^{\I \infty} \exp(\frac{1}{3} z^3 - w z) dz,\\
y_2(w)&:=\E^{-\frac{2\pi \I}{3}} \Ai(\E^{\frac{-2\pi\I}{3}}w),\\
y_3(w)&:=\E^{\frac{2\pi \I}{3}} \Ai(\E^{\frac{2\pi\I}{3}}w).
\end{align*}
Functions $y_i(z)$ are entire functions and connected with each other by the relation \dlmf{9.2.12}
\beq\label{idd}
y_1(w)+y_2(w)+y_3(w)=0.
\eeq
Set
\beq \label{DefOmega}
\Omega_1=\{w:\ \arg w\in \left(0,\,\frac{2\pi}{3}\right)\},\ \Omega_2=\{
 \arg w\in \left(\frac{2\pi}{3},\,\frac{4\pi}{3}\right)\},\ \Omega_3=\C\setminus\ol{\{\Omega_1\cup\Omega_2\}}.
 \eeq
The asymptotics of the Airy functions (cf.\ \dlmf{9.7.5,9.7.6}) are
\begin{align}\label{airas}
y_1(w)&= \begin{cases} \frac{1}{2\sqrt\pi w^{1/4}}\E^{-\frac 2 3 w^{3/2}}(1 +O(w^{-3/2})),& w\in\Omega_1,\\
\frac{\I}{2\sqrt\pi w^{1/4}}\E^{\frac 2 3 w ^{3/2}}(1 +O(w^{-3/2})),& w\in\Omega_3,
\end{cases}\\ \label{airas2}
y_2(w)&=-\frac{\I}{2\sqrt\pi w^{1/4}} \E^{\frac 2 3 w^{3/2}}(1 +O(w^{-3/2})),\quad w\in\Omega_1\cup\Omega_2,\\ \label{airas3}
y_3(w)&=-\frac{1}{2\sqrt\pi w^{1/4}}\E^{-\frac 2 3 w^{3/2}}(1 +O(w^{-3/2})), \quad w\in\Omega_2\cup\Omega_3,
\end{align}
and can be differentiated with respect to $w$. Set
\beq \label{DefA1}
\mathcal{A}_1(w):= \sqrt{\pi}\begin{pmatrix}  y_1(w) & y_2(w) \\ - y_1^\prime(w) & - y_2^\prime(w)\end{pmatrix},\quad w\in \Omega_1.
\eeq
Then $\det \mathcal A_1(w)=1$ (cf.\ \dlmf{9.2.8}), and by \eqref{airas}, \eqref{airas2} we have the correct normalization \eqref{normairy} in $\Omega_1$. Furthermore, functions $\mathcal{A}_i(w),\quad i = 2,3 \quad$ can be found as next $\mathcal{A}_2(w): = \mathcal{A}_1S_2(w),\quad \mathcal{A}_3(w): = \mathcal{A}_2S_3(w)$. And we should check two things: the condition \eqref{normairy} and the monodromy condition $\mathcal{A}_3(w)S_1 = \mathcal{A}_1(w)$.
So, by \eqref{idd}
\begin{align*}
\mathcal A_2(w)&= \mathcal A_1(w)S_2= \sqrt{\pi}\begin{pmatrix}  y_1(w) & y_2(w) \\ - y_1^\prime(w) & - y_2^\prime(w)\end{pmatrix}\begin{pmatrix}  1&\mathcal -1\\0& 1\end{pmatrix}\\
    &=\sqrt{\pi}\begin{pmatrix}-y_3 (w) & y_2(w)\\   y_3^\prime(w) & - y_2^\prime(w)\end{pmatrix},
\end{align*}
\begin{align*}
\mathcal A_3(w)&= \mathcal A_2(w)S_3= \sqrt{\pi}\begin{pmatrix}-y_3 (w) & y_2(w)\\   y_3^\prime(w) & - y_2^\prime(w)\end{pmatrix}\begin{pmatrix} 1& 0\\1 & 1\end{pmatrix}\\
    & = \sqrt{\pi}\begin{pmatrix} - y_3 (w) & -  y_1(w)\\  y_3^\prime(w) &  y_1^\prime(w)\end{pmatrix}.
\end{align*}
In virtue of \eqref{airas} -- \eqref{airas3} matrix $\mathcal A_i(w),\quad i=1,2,3$ obeys the normalization \eqref{normairy} in    $\Omega_i$. For the monodromy condition we have
\begin{align*}
\mathcal{A}_3(w)S_1& = \sqrt{\pi}\begin{pmatrix} - y_3 (w) & -  y_1(w)\\  y_3^\prime(w) &  y_1^\prime(w)\end{pmatrix}\begin{pmatrix} 0& 1\\-1 & 1\end{pmatrix}\\
    & = \sqrt{\pi}\begin{pmatrix}  y_1(w) & y_2(w) \\ - y_1^\prime(w) & - y_2^\prime(w)\end{pmatrix}=\mathcal{A}_1(w).
\end{align*}
So, $\mathcal A(w) = \mathcal A_i(w),\quad w \in \Omega_i \quad i=1,2,3\quad$ is the solution we looked for.

So a result of our previous investigations can be formulated as the following
\begin{theorem}
The matrix solution of the jump problem
\beq\label{parini}
M^{\mathrm{par}}_+(z)=M^{\mathrm{par}}_- (z)v^{\mathrm{par}}(z), \quad z \in \mathcal B\cap(\Sigma_1\cup\Sigma_2\cup \Sigma_3)
\eeq
is given by
\[
M^{\mathrm{par}}(z) = M_0 [\gamma(w)]^{\sigma_3} \mathcal A_j(w(z)) \E^{\frac{2}{3} w^{3/2}(z)\sigma_3},\quad w(z) \in \Omega_j, \quad j = \ol{1,3}\, ,
\]
where $M_0 ,\, \gamma(w),\, w(z)$ defined by \eqref{DefGamma}, \eqref{wz0}, respectively. And $\Omega_i ,\, \mathcal A_1$ are given by \eqref{DefOmega}, \eqref{DefA1} and $\mathcal A_j= \mathcal A_{j-1}S_j,\quad j=2,3.$
\end{theorem}

\section{The conclusion of the asymptotic analysis} \label{asymptotics}
In this section we are going to give substance to the proposition that the solution $m^{(3)}$ of the RH problem \eqref{v3} is well approximated in some sense by $\ra M^{\mathrm{par}}.$ Set
\beq\label{m3mhatrel}
\hat m(z)=m^{(3)}(z) (M^{\text{as}}(z))^{-1},
\quad M^{\text{as}}(z):=\begin{cases} M^{\mathrm{par}}(z), & z\in\mathcal B,\\
 M^{\mathrm{mod}}(z), & z\in\C\setminus\mathcal B.
\end{cases}
\eeq
The vector function $\hat m$ must solve the jump problem
\beq\label{rhpmhat1}
\hat m_+(z)=\hat m_-(z)\hat v(z),
\eeq
where
\beq\nn
\hat v(z)=\begin{cases}M^{\mathrm{par}}_-(z)v^{(3)}(z)(M^{\mathrm{par}}_+(z))^{-1}, & z\in\Sigma_{\mathcal B},\\
M^{\mathrm{par}}(z)(M^{\mathrm{mod}}(z))^{-1}, & z\in \partial \mathcal B,\\
M_-^{\mathrm{mod}}(z)v^{(3)}(z)(M_+^{\mathrm{mod}}(z))^{-1}, & z\in \breve\Gamma\setminus(\Sigma_{\mathcal B} \cup \partial \mathcal B)\end{cases}
\eeq
and the symmetry and normalization conditions:
\beq\label{mhatcond}
\hat m(z^{-1})\hat m(z)\si_1, \quad \hat m_1(0)>0,\quad   \hat m_1(0)\hat m_2(0) = 1.
\eeq
Sets $\Sigma_{\mathcal B}, \mathcal B$ etc. defined by \eqref{parsets} Consider $W(z)=\hat v(z) -\id$.
\beq\label{wdef}
W(z)= \begin{cases}M^{\mathrm{par}}_-(z)\left(v^{(3)}(z)-v^{\mathrm{par}}(z)\right)(M^{\mathrm{par}}_+(z))^{-1}, & z\in\Sigma_{\mathcal B},\\
M^{\mathrm{par}}(z)\,(M^{\mathrm{mod}}(z))^{-1} -\id, & z\in \partial\mathcal B,\\
M_-^{\mathrm{mod}}(z)(v^{(3)}(z) - v^{\mathrm{mod}}(z))(M_+^{\mathrm{mod}}(z))^{-1}, & z\in \Sigma\setminus \Sigma_{\mathcal B},\\
M_-^{\mathrm{mod}}(z)(v^{(3)}(z)-\id)(M_+^{\mathrm{mod}}(z))^{-1}, &  z\in \breve\Gamma\setminus(\Sigma_{\mathcal B} \cup \partial \mathcal B\cup \Sigma).\end{cases}
\eeq
Where $v^{\mathrm{mod}}(z) = \begin{pmatrix}0 & -\mathcal R(-1)\\ \mathcal R(-1) & 0 \end{pmatrix}$ and $v^{\mathrm{par}}(z)=\E^{-t g_-(z)\sigma_3}\, S\, \E^{t g_+(z)\sigma_3},$ i.e.
\beq\nn
v^{\mathrm{par}}(z)=\begin{cases}
\begin{pmatrix}  0&-\mathcal R(-1) \\ \mathcal R(-1)& \E^{-2tg_+(z)}\end{pmatrix}, & z\in\Sigma_1,\\[3mm]
\begin{pmatrix}  1&0\\ -\mathcal R(-1)\E^{2tg(z)}& 1\end{pmatrix}, & z\in\Sigma_2,\\[3mm]
\begin{pmatrix}  1&\mathcal R(-1)\E^{-2tg(z)}\\0& 1\end{pmatrix}, & z\in\Sigma_3.\end{cases}
\eeq
Also recall that
\beq\nn
 v^{(3)}(z)=\left\{\begin{array}{ll}
\begin{pmatrix}
0 & - \mathcal R(-1) \\
\mathcal R(-1) & \frac{\ti d_+(z)}{\ti d_-(z)} \E^{-2tg_+(z)}
\end{pmatrix},  z \in \Sigma,\\
\begin{pmatrix}
1 & 0 \\
-\tilde d^{-2}(z) \mathcal R(z) \E^{2 t g(z)} & 1
\end{pmatrix},  z \in \mathcal C,\\
\begin{pmatrix}
1 & \ti d^2(z) \ol{\mathcal R(z)} \E^{- 2 t g(z)} \\
0 & 1
\end{pmatrix},  z \in \mathcal C^*.
\end{array}\right.
\eeq

Now let's make some observations about matrix $W(z)$.
 \begin{itemize}
 \item
 In virtue of \eqref{v3vmodnorm} and boundness of $M^{\mathrm{mod}}(z)$ on the part $\ti \Sigma \setminus \mathcal B$ we have
 \[ \|W\|_{L^\infty\left(\ti \Sigma\setminus\Sigma_{\mathcal B}\right)}\leq C\E^{-t U}\]
 \item
 The matrix $v^{(3)}(z)-v^{\mathrm{par}}(z)$ has only one non-zero entry on each part of contour $\Sigma_\mathcal B$. We denote this entry by $u(z)$:
 \end{itemize}
\[
u(z)=\begin{cases}(\frac{d_+(z)}{d_-(z)}-1)\E^{-2tg_+(z)}, & z\in\Sigma_1,  \\
(d(z)^{-2}\mathcal R(z) - \mathcal R(-1))\E^{2tg(z)}, & z\in\Sigma_2,\\
(\mathcal R(-1)- d(z)^2\ol{\mathcal R(z)})\E^{-2tg(z)}, & z\in\Sigma_3.   \end{cases}
\]
As $g(z)=\re g(z)\quad \mbox{on} \quad \Sigma_{\mathcal B}$ from \eqref{dlim} and \eqref{change} follows
\beq\label{uz0}
u(z)=\left( C_j I(z_0)\sqrt{|z-z_0|} \right) \E^{-2tg(z)}+ O(z-z_0)\E^{-2tg(z)},\ z\in\Sigma_j.
\eeq
Now, consider the function $ I(z_0)$ defined by \eqref{asymp9}
\[I(z_0)=\frac{\sqrt{1-z_0^2}}{2\pi\I} \int_{z_0}^{\ol z_0} \frac{(\log\mathcal R(s)-\log\mathcal R(z_0))\,ds}{[\sqrt{(z_0 -s)(z_0s-1)}]_+ (s-z_0)}.\]
Denote $f(z_0,z,s):=\frac{\log\mathcal R(s)-\log\mathcal R(z_0)}{(s-z)\sqrt{z_0s-1}_+}$.
\[\int \limits_{z_0}^{-1}\frac{f(z_0,z,s)}{\sqrt{z_0-s}_+}ds = 2f(z_0,z,s)\sqrt{z_0-s}_+ \bigg|_{z_0}^{-1} - \int \limits_{z_0}^{-1} f^\prime_s (z_0,z,s)\sqrt{z_0-s}_+ds,\]
where the last integral can be computed.

On the arc$(-1,\ol z_0)$ the function $\frac{1}{\sqrt{z_0-s}_+}$ does not have any singularities and is differentiable. On this arc the function $h(z_0,z,s):=\frac{\log\mathcal R(s)-\log\mathcal R(z_0)}{(s-z)\sqrt{z_0-s}_+}\,$ is differentiable too.
\[\int \limits_{-1}^{\ol z_0}\frac{h(z_0,z,s)}{\sqrt{z_0s-1}_+}ds = 2h(z_0,z,s)\sqrt{z_0s-1}_+ \bigg|_{-1}^{\ol z_0} - \int \limits_{-1}^{\ol z_0}h^\prime_s (z_0,z,s)\sqrt{z_0s-1}_+ds,\]
and the last integral can be computed too. This consideration shows us that the function $I(z_0$ is differentiable as $\xi \in \mathcal I$.

Next, let $(z_0 + C_j^2\rho_j)$ be the end points of the contours $\Sigma_j$. Recall that $\rho_j\geq\rho\geq\frac{\sqrt{2\epsilon}}{4}$.
Then
\beq \nn
\int \limits_{\Sigma_j} u(z)dz=[s = z-z_0]= C_j\tilde I(z_0) \int \limits_{0}^{\rho_j} s^{\frac{1}{2}} \E^{-2tg(s)}ds + O(t^{-4/3}) .
\eeq
By lemma \ref{gexp} in a vicinity of $z_0$ we have
\[\int  s^{\frac{1}{2}} \E^{-2tg(z)}ds =\int s^{\frac{1}{2}} \E^{-2C(\vphi_0)ts^{3/2}}\E^{-2t(g(s)-s^{3/2})}ds, \]
and
\[\E^{-2t(g(s)-s^{3/2})} \sim \E^{t\ti C s^{5/2}}= \sum \limits_{k=0}^{\infty}(-1)^\frac{\ti C^k t^k s^{5k/2}}{k!}.\]
\[\int s^{\frac{1}{2}} \E^{-2tg(z)}ds \sim \left[y:=ts^{2/3},\quad s^{5k/2} = \left(\frac{y}{t}\right)^{5k/3}\right]\sim \frac{3}{2}\int  \sum \limits_{k=0}^{\infty}\frac{\ti C^k}{k!}y^{\frac{5k}{3}}t^{-1-\frac{2k}{3}}\E^{-Cy}dy.\]
\[\int y^{\frac{5k}{3}}t^{-1-\frac{2k}{3}}\E^{-Cy}dy= O(t^{-\alpha}),\quad \alpha >1,\quad \mbox{when}\,\, k>0,\]
so we conclude
\[
\int \limits_{\Sigma_j} u(z)dz = \frac{H_j(z_0)}{t} + O(t^{-\alpha}),\quad \alpha >1,\]
where the function $H_j(z_0)$ is differentiable and uniformly bounded with respect to $\xi \in \mathcal I$. We also have the following estimete
\beq\nn
\|u\|_{L^1(\Sigma_{\mathcal B})}=O(t^{-1}), \quad z \in \Sigma_\mathcal B.
\eeq
Moreover, the change of the variable of integration $y := \big(2C(\vphi_0)t\big)^{2/3}s\quad$ implies the similar estimate in $L^p(\Sigma_{\mathcal B})$
\beq\nn
\|u\|_{L^p(\Sigma_{\mathcal B})}=O(t^{-\frac{2+p}{3p}}), \quad z \in \Sigma_\mathcal B.
\eeq
Using the same arguments and taking into account that the matrix entries \newline${[ M_-^{\text{par}}]_{mk}(z)[( M^{\text{par}}_+)^{-1}]_{pq}(z)}$, $m,k,p,q = \ol{1,2}$ are bounded for $z\in\Sigma_{\mathcal B}$
uniformly with respect to $\xi\in\mathcal I$, and using \eqref{wz0} and \eqref{uz0}, we get:
\beq\label{estu}
\int \limits_{\Sigma_{\mathcal B}} u(z)[M_-^{\text{par}}]_{mk}(z)[(M^{\text{par}}_+)^{-1}]_{pq}(z)dz
   =\frac{h_{m,k,p,q}(z_0)}{t} +O(t^{-4/3}).
\eeq
The functions $h_{m,k,p,q}(z_0)$ are  bounded with respect to $\xi\in\mathcal I$ \eqref{estu} implies
\beq\label{uses}
\int \limits_{\Sigma_\mathcal B} \,W(z)dz= \frac{F_\mathcal B(z_0)}{t} + O(t^{-4/3}),
\eeq
where the matrix $F_\mathcal B(z_0)$ is bounded for $\xi\in\mathcal I$. From \eqref{wdef}, \eqref{mpar}, \eqref{beta} and \eqref{uz0} it follows that the function $W(z)$ does not have a singularity at the point $z_0$. Next we note that the function $g(z)$ \eqref{g1} has a continuous derivative with respect to $z_0$ for $\xi\in\mathcal I$, so the function $F_\mathcal B(z_0)$ is differentiable.
The boundedness allows us to get the following estimates
\beq\nn
\|W\|_{L^1(\Sigma_{\mathcal B})}=O(t^{-1}), \quad  \|W\|_{L^\infty(\Sigma_{\mathcal B})}=O(t^{-1/3}).
\eeq
Moreover, from \eqref{wdef} it follows that
\beq\label{uses2}
 \int \limits_{\partial \mathcal B}W(z)dz = \frac{F_{\partial\mathcal B}(z_0)}{t\,\rho^{1/2}} + O(t^{-4/3}),
\eeq
where the matrix $F_{\partial\mathcal B}(z_0)$ have the same properties as $F_\mathcal B(z_0)$, since
\[W(z)= \frac{1}{72t\,g(z)}\begin{pmatrix}-7 & 7\\ 5& -5\end{pmatrix}+ O(t^{-2}).\]
The last statement follows from \dlmf{9.7.5,9.7.6} and \eqref{mparmod}.

Next, the matrix $M^{\text{mod}}(z)$ and its inverse are bounded with an estimate $O(\rho^{-1/4})$ on  the remaining part of the contour $\tilde\Sigma$. Using \eqref{wdef}, \eqref{v3vmodnorm} and \eqref{brevexi} we conclude
\beq\label{remW}
\int \limits_{\tilde\Sigma\setminus(\Sigma_{\mathcal B}\cup \partial \mathcal B)}\,W(z) dz=\tilde F_{mod}(z_0,\rho, t),\quad \|\tilde F_{mod}(z_0,\rho, t)\|\leq C\rho^{-1/4}\E^{-\frac{\rho t}{2}},
\eeq
%Ðàâíîìåðíî ïî $\xi$ ^
where the matrix norm of $F_{mod}(z_0,\rho,t)$  is uniformly bounded with respect to $\xi$ and $\rho$ for $t\in[T_0,\infty)$ and $\xi\in\mathcal I$.  We can conclude also
\[
\|W(z)\|_{L^1(\tilde\Sigma\setminus(\Sigma_{\mathcal B}\cup \partial \mathcal B))}\leq O(\E^{-\epsilon t}),\quad \|W(z)\|_{L^\infty(\tilde\Sigma\setminus(\Sigma_{\mathcal B}\cup \partial \mathcal B))}\leq O(\E^{-\epsilon t}).
\]
The contour $\ti \Sigma$ does not contain the point $z = 0$, so multiplying the error matrix $W(z)\quad$ by $\quad z^p,\quad p \in \R \quad$ preserves previous estimates:
\beq\nn
\int \limits_{\Sigma_\mathcal B}\left|z^p\,W(z)\right|dz\le \|z^p\|_{C(\ti \Sigma)}\,\|W(z)\|_{L^1(\breve\Gamma).}
\eeq
The previous estimates allow us to formulate the following
\begin{lemma}\label{WLp}
The following estimates holds uniformly for $\xi \in \mathcal I,\, 1\leq p \leq \infty$
\[\|W\|_{L^p(\Gamma)} = O\big(t^{-\frac{1}{3}-\frac{2}{3p}}\big).\]
\end{lemma}
To solve the RH problem \eqref{rhpmhat1} -- \eqref{mhatcond} we need to find out the behavior of the vector $m^{(3)}(z)$ at any "appropriate" point. In our case we choose $z=1$ as such a point. Observe that at the point $z=1$ the vector-function $m(z) = \begin{pmatrix} T(z,t) \psi_{\ell}(z,n,t) z^n   & \psi(z,n,t)  z^{-n} \end{pmatrix}$ has the following structure
\[m(1)= \begin{cases} \begin{pmatrix}0, & \ti\eta \end{pmatrix}& \mbox{in a non-resonant case},\\
        \begin{pmatrix} 2\hat\eta, &\hat\eta\end{pmatrix}& \mbox{in a resonant case}, \end{cases}\]
where $\ti\eta(n,t)$ and $\hat\eta(n,t)$ are bounded functions with respect to all arguments. It can be easily checked that after steps {\it 1} -- {\it 4} from the section \ref{sec:steps} at the point $z=1$ the vector-function in both resonant and non-resonant cases $m^{(3)}$ has the next structure
\beq\label{m31}
    m^{(3)}(1) = \rn,
\eeq

where $\eta = \eta(n,t)$ is a bounded function with respect to all arguments.

\begin{lemma}\label{mhatsol}
The solution of the RH problem \eqref{rhpmhat1}--\eqref{mhatcond} has the following form
\beq\nn
\hat m(z)=\ra + \frac{f(\xi,\rho)}{t}+z\frac{\hat f(\xi,\rho)}{t} + \displaystyle o(t^{-1})\displaystyle o(z),
\eeq
where $\alpha = \alpha(\xi)$ defined by\eqref{mmod}.
\end{lemma}
\begin{proof}
Let $\mathfrak C$ denote the Cauchy operator associated with $\breve\Gamma$ and with kernel $\Omega(z,s):=\frac{1}{2}\left(\frac{s+z}{s-z} - \frac{s+1}{s-1} \right)$:
\[
(\mathfrak C h)(z)=\frac{1}{2\pi\I}\int \limits_{\breve\Gamma}h(s)\Omega (z,s)\frac{ds}{s}, \qquad k\in\C\setminus\breve\Gamma,
\]
where $h= \begin{pmatrix} h_1 & h_2 \end{pmatrix}\in L^2(\breve\Gamma)\cap L^\infty(\breve\Gamma)$.
Let  $\mathfrak C_+ f$ and $\mathfrak C_- f$ be its non-tangential limiting values from the left and right sides of $\breve\Gamma$, respectively.
These operators will be bounded with bound depending on the contour, that is on $z_0$. However, since we can choose our contour
scaling invariant at least locally, scaling invariance of the Cauchy kernel implies that we can get a bound which is uniform on compact sets.

Using the Cauchy operator $\mathfrak C$ let's introduce the operator $\mathfrak C_{W}:L^2(\breve\Gamma)\cap L^\infty(\breve\Gamma)\to
L^2(\breve\Gamma)$ by $\mathfrak C_{W} f=\mathfrak C_-(f W)$, where $W$ is our error matrix \eqref{wdef}.
Then,
\beq\label{cwest}
\|\mathfrak C_{W}\|_{L^2(\breve\Gamma)\to L^2(\breve\Gamma)}\leq C\|W\|_{L^\infty(\breve\Gamma)}\leq O(t^{-1/3}).
\eeq

Utilizing the Neumann series representation
\[
\left(\id - \mathfrak C_{W}\right)^{-1}=\sum \limits_{j=0}^{\infty}\mathfrak C_{W}^j
\]
we obtain
\[
\|(\id - \mathfrak C_{W})^{-1}\|_{L^2(\breve\Gamma)\to L^2(\breve\Gamma)}\leq \sum \limits_{j=0}^{\infty}\| \mathfrak C_{W} \|^j_{L^2(\breve\Gamma)\to L^2(\breve\Gamma)}= \frac{1}{1-\|\mathfrak C_{W}\|_{L^2(\breve\Gamma)\to L^2(\breve\Gamma)}},
\]
and by \eqref{cwest}
\beq\label{cwinv}
\|(\id - \mathfrak C_{W})^{-1}\|_{L^2(\breve\Gamma)\to L^2(\breve\Gamma)}\leq \frac{1}{1-O(t^{-1/3})}
\eeq
for sufficiently large $t$. Now, for $t\gg 1$ we can define next vector function
\[
\mu(z) =\rn + (\id - \mathfrak C_{W})^{-1}\mathfrak C_{W}\big(\rn\big)(z),
\]
here the component $\eta$ is defined by \eqref{m31}.
\beq\label{MuEtaLpEst}
\|\mu(z) - \rn\|_{L^2(\breve\Gamma)} \leq \|(\id - \mathfrak C_{W})^{-1}\|_{L^2(\breve\Gamma)\to L^2(\breve\Gamma)} \|\mathfrak C_{-}\|_{L^2(\breve\Gamma)\to L^2(\breve\Gamma)} \|W\|_{L^2(\breve\Gamma)}.
\eeq
Then by \eqref{cwinv} and lemma \ref{WLp}
\beq\nn
\|\mu(z) - \rn\|_{L^2(\breve\Gamma)} \leq \displaystyle O(t^{-2/3}).
\eeq
The solution of the RH problem \eqref{rhpmhat1} can be expressed by
\[
\hat m(z)=\rn +\frac{1}{2\pi\I}\int \limits_{\breve\Gamma}\mu(s) W(s)\Omega (z,s)\frac{ds}{s}.
\]
Rewriting this equation we can obtain
\[
\hat m(z)=\rn + \frac{1}{2\pi\I}\int \limits_{\breve\Gamma}\rn W(s)\Omega (z,s)\frac{ds}{s} + G(z),
\]
where
\[
G(z):= \frac{1}{2\pi\I}\int \limits_{\breve\Gamma}\left(\mu(s)-\rn\right) W(s)\Omega (z,s)\frac{ds}{s}.
\]
Further, decomposition of the function $\Omega(z,s)$ in a vicinity of the point $z=0$ has the following form
\[
\Omega(z,s) = \frac{1}{2}\left(1 - \frac{s+1}{s-1} + \frac{2}{s}z\right) + \displaystyle O(z^2) =-\frac{1}{s-1} + \frac{1}{s}z + \displaystyle O(z^2)s^{-2}.
\]
Observe that our contour $\breve\Gamma$ does not contain points $z = 0$ and $z=1$. Consequently, on this contour function $\Omega(z,s)$ is bounded with respect to $s$ (and even uniformly bounded). Using this fact and \eqref{MuEtaLpEst} allows us to estimate the function $G(z)$:
\[
|G(z)|\leq\|W\|_{L^2(\breve\Gamma)}\|\mu (k)- \rn\|_{L^2(\breve\Gamma)}
\left(1 + \displaystyle O(z) \right)=\displaystyle O(t^{-4/3}) + \displaystyle O(z)\displaystyle O(t^{-4/3}).
\]
\begin{align*}
\frac{1}{2\pi\I}\int \limits_{\breve\Gamma}\rn W(s)\Omega (z,s)\frac{ds}{s} =& \rn \frac{F_1(z,t,\xi)}{t} +\rn\frac{F_2(z,t,\xi)}{t}z\\ & + \displaystyle{O(z^2)O(t^{-1})},
\end{align*}
here matrices $F_i(z,t,\xi),\quad i=1,2$ are uniformly bounded with respect to all arguments.
So, using \eqref{uses}, \eqref{uses2} and \eqref{remW} for $z\to 0$ we have
\[\hat m(z)=\rn + \frac{f(\xi,\rho)}{t}+z\frac{\hat f(\xi,\rho)}{t} + \displaystyle o(t^{-1})\displaystyle o(z).\]
Furthermore, using this equation and \eqref{m3mhatrel} we get
\[
m^{(3)}(z)  = \hat m(z)M^{\text{as}}(z) = \rn M^{\text{mod}}(0) + M^{\text{mod}}(0)O(t^{-1}),\quad \mbox{as} \quad z \to 0.
\]
\noprint{
Recall that $m^{(3)}(z)$ satisfies (by \ref{m3th}) the normalization condition \eqref{symcond}. Then by \eqref{Mmodinfty}, \eqref{mmod} and the fact that $\beta(z^{-1})=\beta^{-1}(z)$ for sufficiently large $t$ we conclude
\[\eta(n,t) = \alpha(\xi) .\]
}
So,
\[\hat m(z) \approx \rn  + O(t^{-1}),\quad z \in \ol \C,\]
and by \eqref{m3mhatrel}, \eqref{Mmod}
\[m^{(3)}(z) = \hat m(z)M^{\mathrm{mod}}(z) + O(t^{-1}) = \rn M^{\mathrm{mod}}(z) + O(t^{-1}).\]
The function $m^{(3)}(z)$ satisfies the normalization condition \eqref{eq:normcond} and the symmetry condition \eqref{sc} and $M^{\mathrm{mod}}(z)$ satisfies the symmetry condition: $M^{\mathrm{mod}}(z^{-1})=\sigma_1 M^{\mathrm{mod}}(z)\sigma_1$. So for sufficiently large $t$ at the point $z=0$ we get
\begin{align}\nn
&m^{(3)}(0)  = \begin{pmatrix} m^{(3)}_1, & m^{(3)}_2 \end{pmatrix} = \rn\sigma_1 M^{\mathrm{mod}}(\infty)\sigma_1,\\
&m^{(3)}_1 \cdot m^{(3)}_2  = 1, \nn
\end{align}
and as follows from \eqref{mmod}
\[ \eta = \alpha(\xi) + O(t^{-1}).\]
\end{proof}
\begin{remark}
From the discussion on the differentiability of the functions $F_{\mathcal B},\,F_{\partial \mathcal B}\, \mbox{and}\, \ti  F_{mod}$ it follows that the functions $f(\xi,\rho)$ and $\hat f(\xi,\rho)$ are differentiable with respect to $\xi \in \mathcal I$.
\end{remark}
Since by \eqref{conn} $\quad m^{(3)}(z)=m(z)\left[\tilde d(z) P(z) \E^{t\left(\Phi(z) - g(z)\right)}\right]^{-\sigma_3}$ , then from \eqref{m3mhatrel} we have
\beq\nn
m(z)=\hat m(z)M^{\text{mod}}(z)\left(\tilde d(z) P(z) \E^{t\left(\Phi(z) - g(z)\right)}\right)^{\sigma_3}.
\eeq
The last equation with lemma \ref{mhatsol}  and theorem 5.1 \cite{EM} allow us to formulate the following
\begin{theorem}
For arbitrary small number $\eps >0$ the following asymptotics for the solution of the Toda lattice \eqref{tl} -- \eqref{decay} as $t \to \infty$ are valid in the region $\quad \eps t \le n \le (1-\eps)t$:
\[
a(n,t) = \frac{n}{2t}  + \displaystyle O(t^{-1}),
\]
\[
b(n,t) = 1 - \frac{n}{t} + \displaystyle O(t^{-1}),
\]
and here terms $\displaystyle O(t^{-1})$ does not contains terms from the solution of the parametrix problem.
\end{theorem}

%\begin{align*}
%\frac{\mathcal A(\xi)}{\mathcal A(\ti\xi)}& = -K^\prime(\xi)+\displaystyle O(t^{-1}),\\
%\mathcal B(\xi) -\mathcal B(\hat \xi)& = -\frac{k^\prime(\xi)}{2}+\displaystyle O(t^{-1}).
%\end{align*}
\appendix
\section{Solution of the problem I -- V}
The existence of the solution of the RH problem I -- V was proved in  \cite{dkkz}(2.13 -- 2.15, 2.17, 2.18). To prove the uniqueness we will need the following
\begin{lemma}\label{SolsDep}
Let $f(z) = \begin{pmatrix}f_1(z) & f_2(z) \end{pmatrix}$ and  $g(z) = \begin{pmatrix}g_1(z) & g_2(z) \end{pmatrix}$ be two solutions of the RH problem I -- V. Then $f(z) = c(z)g(z)$, where $c(z)$ is a scalar function without jumps on the contour $\T \cup I \cup I^*$ and $\lim \limits_{z\to\infty}c(z) = 1$.
\end{lemma}
\begin{proof}
Consider the function $S(z)$ defined by
\[S(z):=\begin{pmatrix}f_1(z)&f_2(z) \\ g_1(z)& g_2(z) \end{pmatrix}.\]
Then
\beq\label{Srhp}
S_+(z) = S_-(z)v(z).
\eeq
Furthermore, let $s(z):=\det S(z)$. Since $\det v(z) = 1$ by \eqref{Srhp} $s_+(z) = s_-(z),$ i.e. the function $s(z)$ does not have a jump along the contour. Denote
\beq\label{SolDiff}
h(z) = f(z) - g(z).
\eeq
By the pole condition \eqref{polecond}
\[
\res \limits_{z_j}\left[f_1(z)g_2(z) - g_1(z)f_2(z) \right] = -z_j\ga_j\E^{2 t\Phi(z_j)}\left[f_2(z_j)g_2(z_j) - g_2(z_j)f_2(z_j)\right]\equiv0,
\]
this means, that the function $s(z)$ does not have poles at eigenvalues in non-resonant case. In a resonant case $s(z) = O\big( (z - p_j)^{-1/2}\big),\quad p_j \in \{ q_1,q_2, -1, 1, q_1^*, q_2^*\}$, so $s(z)$ is a bounded function at these points.
In virtue of the normalization \eqref{eq:normcond} and the symmetry \eqref{sc} conditions the function $s(z)$ is bounded at $\infty$. So, $s(z)$ is bounded holomorphic function, then by Liouville's theorem $s(z) \equiv const$. Using the symmetry condition once more, we have \[s(z^{-1}) = f_2(z)g_1(z) - f_1(z)g_2(z) \equiv -s(z).\]
Then $s(1) = -s(1)$ and $s(z)\equiv 0$. Consequently $f(z) = c(z)g(z)$, where $c(z)$ has no jumps on the jump contour. Moreover from the normalization condition \eqref{eq:normcond} follows $\lim \limits_{z\to\infty}c(z) = 1$.
\end{proof}

\begin{proof}
In virtue of Lemma \ref{SolsDep} showing that the associated vanishing problem, where the normalization
condition \eqref{eq:normcond} is replaced by:
\[
|h_1(\infty)|+|h_2(\infty)| = 0,
\]
here the vector-function $h(z) = \begin{pmatrix} h_1(z) & h_2(z) \end{pmatrix}$ defined by \eqref{SolDiff}, has only the trivial solution is enough to prove the uniqueness. Consider the scalar bounded (meromorphic) function
\[ F(z):= h_1(z)\ol{h_1(\ol{z^{-1}})} +  h_2(z)\ol{h_2(\ol{z^{-1}})}
\]
By Cauchy's residue theorem
\beq\nn
\int \limits_{\mathcal C_\eps} \, F(z)\frac{\I dz}{z} = 2\pi\I\left(\sum \limits_{z_j}\frac{\I}{z_j}\res \limits_{z_j}F(z) + \I\res \limits_{z=0}\frac{F(z)}{z}\right).
\eeq
\begin{figure}[h]
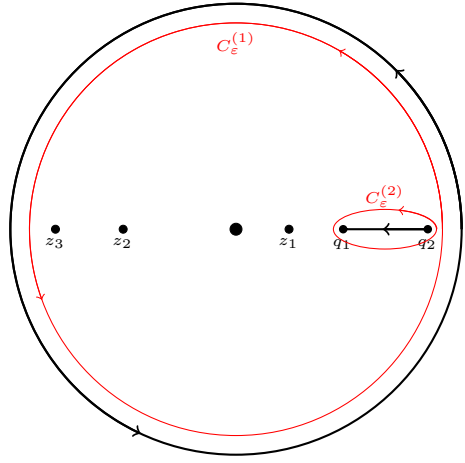

\tikz[scale=1.5]
{
\draw[thick] circle(2);
\draw[thick] (2,0)arc(0:45:2)[->];
\draw[thick] (2,0)arc(0:245:2)[->];
\filldraw[black] circle(1.5pt)
                 (0.47,0)circle(1pt)node[below]{\small$_{z_1}$}
                 (-1,0)circle(1pt)node[below]{\small$_{z_2}$}
                 (-1.6,0)circle(1pt)node[below]{\small$_{z_3}$};
\filldraw[thick ] (0.95,0)circle(0.8pt)node[below]{\small$_{q_1}$}--(1.7,0)circle(0.8pt)node[below]{\small$_{q_2}$};
\draw[thick](1.7,0)--(1.3,0)[->];
\draw[red] circle(1.83) ;
\draw [red](1.83,0)arc(0:60:1.83)[->];
\draw[red] (1.83,0)arc(0:200:1.83)[->];
\draw[red] (0,1.83)node[below]{\tiny$C_\varepsilon ^{(1)}$};
%\draw[red] (1.8,0)arc(0:90:13pt and 5pt)[->];
%\draw[red] (1.8,0)arc(90:-270:13pt and 5pt);
\draw[red] (1.32,0)node[above=5]{\tiny$C_\varepsilon ^{(2)}$}ellipse(13pt and 5pt);
\draw[red] (1.78,0)arc(0:75:13pt and 5pt)[->];
}
\caption{A part of the jump contour for the problem I--V}\label{UniqCont}
\end{figure}
From the vanishing condition follows $F(0) = 0$, then
\[
\int \limits_{\mathcal C_\eps} \, F(z)\frac{\I dz}{z} = -2\pi\sum \limits_{z_j}\frac{1}{z_j}\res \limits_{z_j}F(z).
\]
First, we consider the integral of $F(z)$ on the part $\mathcal C^{(1)}_\eps$ of the contour $\mathcal C_\eps$(see Fig.~\ref{UniqCont}), such as $\mathcal C^{(1)}_\eps \to \T, \quad \eps \to 0$. When $|z| \to 1\pm0,\quad F(z) = F_\pm(z)$.
\[
F_+(z) = h_{1,+}(z)\ol{h_{1,-}(z)} + h_{2,+}(z)\ol{h_{2,-}(z)}
\]
From the statement of the initial RHP we have
\begin{align*}
\begin{pmatrix}h_1(z) & h_2(z)\end{pmatrix}_+ &= \begin{pmatrix}h_1(z) & h_2(z)\end{pmatrix}_- \begin{pmatrix}0 & - \ol{R}(z) \E^{- 2 t \Phi(z)} \\R(z) \E^{2 t \Phi(z)} & 1\end{pmatrix},\\
    h_{1,+}(z) &= h_{2,-}(z) R(z) \E^{2 t \Phi(z)},\\
    h_{2,+}(z) & = - h_{1,-}(z) \ol{R}(z) \E^{- 2 t \Phi(z)} + h_{2,-}.
\end{align*}
Then
\begin{align*}
F_+(z)& = \ol{h_{1,-}(z)}h_{2,-}(z) R(z) \E^{2 t \Phi(z)} - h_{1,-}(z)\ol{h_{2,-}(z)} \ol{R(z)} \E^{- 2 t \Phi(z)} + h_{2,-}(z)\ol{h_{2,-}(z)}\\
&=\ol{h_{1,-}(z)}h_{2,-}(z) R(z) \E^{2 t \Phi(z)} - \ol{\ol{h_{1,-}(z)}h_{2,-}(z) R(z) \E^{2 t \Phi(z)}} + |h_{2,-}(z)|^2\\
&=2\I\im\ol{h_{1,-}(z)}h_{2,-}(z) R(z) \E^{2 t \Phi(z)}  + |h_{2,-}(z)|^2.
\end{align*}
\begin{align*}
\int \limits_{\T} \, \re F_+(z)\frac{\I dz}{z} = -\int \limits_{0}^{2\pi} \,|h_{2,-}(\E^{\I \theta})|^2 d  \theta \in \R,\\
2\I\int \limits_{\T} \, \im F_+(z)\frac{\I dz}{z} = -2\int \limits_{0}^{2\pi} \, \im F_+(\E^{\I \theta})d  \theta \in \R.
\end{align*}
Next, we consider the integral on the part $\mathcal C^{(2)}_\eps$, where $\mathcal C^{(2)}_\eps \to I, \quad \eps \to 0$.  Firstly, we note $F(\xi) \to F_{\pm}(z), \quad \mbox{as} \quad \xi \to z \pm 0\I,\quad z\in I$.
\[
F_+(z) = h_{1,-}(z)\ol{h_{2,+}(z)} + h_{2,-}(z)\ol{h_{1,+}(z)}
\]
From the definition of the matrix $v(z)$
\begin{align*}
h_{1,+}(z) &= h_{1,-}(z) + h_{2,-}(z)\chi(z) \E^{2t\Phi(z)},\\
h_{2,+}(z) &= h_{2,-}.
\end{align*}
So
\begin{align*}
F_+(z) &= h_{1,-}\ol{h_{2,-}(z)} + h_{2,-}\ol{h_{1,-}(z)} + h_{2,-} \ol{h_{2,-}(z)\chi(z) \E^{2t\Phi(z)}}\\
&= h_{1,-}\ol{h_{2,-}(z)} + \ol{h_{1,-}\ol{h_{2,-}(z)}} + |h_{2,-}(z)|^2\ol{\chi(z) \E^{2t\Phi(z)}}\\
&=2\re h_{1,-}\ol{h_{2,-}(z)} + \I|h_{2,-}(z)|^2|\chi(z)| \E^{2t\Phi(z)}.
\end{align*}
Analogically, for $F_-(z)$ we have
\[F_-(z) = 2\re h_{2,-}(z)\ol{h_{1,-}(z)} + -\I|h_{2,-}(z)|^2|\chi(z)| \E^{2t\Phi(z)}.\]
Note, that $\quad \re F_-(z) = \re F_+(z),\quad \im F_-(z) = -\im F_+(z)$. Then as $\eps \to 0$ we obtain
\begin{align*} \int \limits_{\mathcal C^{(2)}_\eps}\, F(\xi)\frac{\I d \xi}{\xi} &= \I \int \limits_{q_1}^{q_2}\,\re F_+(z)\frac{dz}{z} -\int \limits_{q_1}^{q_2}\,\im F_+(z)\frac{dz}{z}\\
 &+ \I \int \limits_{q_2}^{q_1}\,\re F_-(z)\frac{dz}{z} -\int \limits_{q_2}^{q_1}\,\im F_-(z)\frac{dz}{z}\\
 &= -\int \limits_{q_1}^{q_2}\,\im F_+(z)\frac{dz}{z} = -\int \limits_{q_1}^{q_2}\, |h_{2,-}(z)|^2|\chi(z)| \E^{2t\Phi(z)}\frac{dz}{z}.
 \end{align*}

 \begin{align*}
 \I \int \limits_{\mathcal C_\eps}\,F(z)\frac{dz}{z} &= -\int \limits_{0}^{2\pi} \,|h_{2,-}(\E^{\I \theta})|^2 d\theta -\int \limits_{q_1}^{q_2}\, |h_{2,-}(z)|^2|\chi(z)| \E^{2t\Phi(z)}\frac{dz}{z}-2\int \limits_{0}^{2\pi} \, \im F_+(\E^{\I \theta})d  \theta.\\
 &=-2\pi\sum \limits_{z_j}\frac{1}{z_j}\res \limits_{z_j}F(z).
 \end{align*}
 From the pole condition \eqref{polecond} we have
 \[
 \res \limits_{z_j}h_1(z)\ol{h_2(\ol{z})} = -z_j \ga_j \E^{2t\Phi(z_j)}|h_2(z_j)|^2
 \]
 and
 \[ \I \int \limits_{\mathcal C_\eps}\,F(z)\frac{dz}{z} = 4\pi\sum \limits_j \ga_j \E^{2t\Phi(z_j)}|h_2(z_j)|^2.\]
 Note that $\int \limits_{0}^{2\pi} \, \im F_+(\E^{\I \theta})d  \theta = 0$ then
\[
 \int \limits_{0}^{2\pi} \,|h_{2,-}(\E^{\I \theta})|^2 d\theta +\int \limits_{q_1}^{q_2}\, |h_{2,-}(z)|^2|\chi(z)| \E^{2t\Phi(z)}\frac{dz}{z} +4\pi\sum \limits_j \ga_j \E^{2t\Phi(z_j)}|h_2(z_j)|^2=0.
 \]

we obtain
\begin{itemize}
\item
$h_2(z_j) = 0$, and consequently $\res \limits_{z_j}h_1(z) = 0$, so the vector function $h(z)$ does not have any poles.
\item
$h_{2}(z)=0,\quad z \in \T \cup  I$, so $h_2(z) \equiv 0,\quad z \in \mathbb{D}$, as it is analytic function and the set $I$ has a condensation point. Then, by the symmetry condition \ref{sc} $h_1(z) \equiv 0, \quad z \in \C \setminus \mathbb{D}$.
\item
From \eqref{sc} we obtain $h_{1}(z) = h_{2}(z^{-1})=0,\quad z \in \mathbb{D} $.
\end{itemize}
Combining these statements we conclude
\[
h(z) = \begin{pmatrix} h_1(z) & h_2(z) \end{pmatrix} =  \begin{pmatrix} 0 & 0 \end{pmatrix}.
\]
\end{proof}

\section{Uniqueness for the model problem}

\begin{lemma}\label{ModProbUniq}
The solution of the model RH problem is unique.
\end{lemma}
\begin{proof}
As follows from proof of the uniqueness for the problem I--V and \ref{SolsDep} to prove the uniqueness of the model problem, it suffices to prove that the associated RH problem with the normalization condition replaced by
\[
|h_1(\infty)|+|h_2(\infty)| = 0,
\]
where vector-function $h(z)$ is a difference between two not identical solutions of the problem \eqref{modRHP}, has only the trivial solution. To this end we consider the following function
\[ F(z):= h_1(z)\ol{h_2(\ol{z^{-1}})}. \]
Note that
\[  F_\pm(z) = h_{1,\pm}(z)\ol{h_{2,\mp}(z)}, \]
From \eqref{vmod} we have
\[ h_{1,+}(z)  = h_{2,-}(z)\mathcal R(-1), \quad  h_{2,+}(z) = - h_{1,-}(z)\mathcal R(-1), \]
then
\[ F_+(z) = \ol{h_{2,-}(z)} h_{2,-}(z)\mathcal R(-1) = \mathcal R(-1)|h_{2,-}(z)|^2, \]
\[ F_-(z) = -h_{1,-}(z)\ol{h_{1,-}(z)}\mathcal R(-1) = -\mathcal R(-1)|h_{1,-}(z)|^2. \]
\begin{figure}[ht]
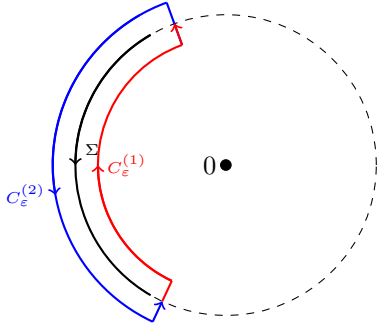

\tikz
{
\draw[dashed] circle(2);
\draw[rotate=120] (2,0)arc(0:60:2)[thick][->]node[above right]{\tiny$\Sigma$};
\draw[rotate=120] (2,0)arc(0:120:2)[thick];
\draw[rotate=120] (2.3,0)arc(0:70:2.3)[blue][->][thick]node[left]{\tiny$C_\varepsilon^{(2)}$};
\draw[rotate=180] [red][<-](1.7,0)node[right]{\tiny$C_\varepsilon^{(1)}$}arc(0:65:1.7)[thick];
\draw[rotate=110] (2.3,0)arc(0:135:2.3)[blue,thick];
\draw[rotate=110](1.7,0)arc(0:135:1.7)[thick,red];
\draw[rotate=110] [thick](2.3,0) -- (1.7,0)[blue];
\draw[rotate=110] [thick](1.7,0) -- (2,0)[red][->];
\draw[rotate=245] [thick](2.3,0) -- (2,0)[->][blue];
\draw[rotate=245] [thick](2,0)--(1.7,0)[red];
\filldraw[black] (0,0)node[left]{$0$}circle(2pt);
}
\caption{The jump contour for the model problem}\label{ModProbCont}
\end{figure}
Since the function $h(z)$ does not have any poles inside the contour $C_\eps$  and taking into account that $C_\eps \to \Sigma ,\quad \eps \to 0\quad$ and that contours $C_\eps^{(1)}$ and $C_\eps^{(2)}$(see Fig. \ref{ModProbCont}) oriented oppositely we get
\[
\int \limits_{\mathcal C_\eps} \, F(z)\frac{\I dz}{z} = \mathcal R(-1)\int \limits_{\Sigma} \left( |h_{2,-}(z)|^2 + |h_{1,-}(z)|^2 \right)\frac{dz}{z} = 0.
\]
Consequently we have $h_{2,-}(z) = h_{1,-}(z) = 0$ and using the symmetry condition we conclude $\begin{pmatrix}h_1(z) & h_2(z) \end{pmatrix}= \begin{pmatrix} 0  & 0 \end{pmatrix}.$
\end{proof}

\end{document}